\documentclass[lettersize,onecolumn]{IEEEtran}
\usepackage{amsmath,amsfonts}
\usepackage{algorithmic}
\usepackage{algorithm}
\usepackage{array}
\usepackage[caption=false,font=normalsize,labelfont=sf,textfont=sf]{subfig}
\usepackage{textcomp}
\usepackage{stfloats}
\usepackage{url}
\usepackage{verbatim}
\usepackage{graphicx}
\usepackage{cite}
\usepackage[numbers]{natbib}
\usepackage{amssymb,amsthm}
\newtheorem{thm}{Theorem}
\newtheorem{lem}{Lemma}

\newtheorem{rem}{Remark}

\newcommand{\zfour}{\mathbb{Z}_4}
\newcommand{\ztwo}{\mathbb{Z}_2}
\newcommand{\f}{\mathbb{F}_{2^m}}
\newcommand{\Tr}{{\rm {Tr}}}
\newcommand{\tr}{{\rm {tr}}}
\newcommand{\rr}{\mathbb{R}}
\newcommand{\cc}{\mathcal{C}}
\newcommand{\tei}{\mathbb{L}}
\newcommand{\bc}{\boldsymbol{c}}
\newcommand{\re}{\mathfrak{Re}}
\newcommand{\im}{\mathfrak{Im}}

\usepackage{arydshln}
\begin{document}
	
	\title{The Lee weight distributions of several classes of linear codes over $\mathbb{Z}_4$}
	
	\author{
		Zhexin Wang\thanks{Z. Wang and X. Zeng are with the Key Laboratory of Intelligent Sensing System and Security (Hubei University), Ministry of Education, the Hubei Key Laboratory of Applied Mathematics, Faculty of Mathematics and Statistics, Hubei University, Wuhan 430062, China. Email: zhexin.wang@aliyun.com, xiangyongzeng@aliyun.com},
		Nian Li\thanks{N. Li is with the Key Laboratory of Intelligent Sensing System and Security (Hubei University), Ministry of Education, the Hubei Provincial Engineering Research Center of Intelligent Connected Vehicle Network Security, School of Cyber Science and Technology, Hubei University, Wuhan 430062, China, and also with the State Key Laboratory of Integrated Services Networks, Xi'an 710071, China. Email: nian.li@hubu.edu.cn},
		Xiangyong Zeng,
		Xiaohu Tang \thanks{X. Tang is with the Information Coding and Transmission Key Lab of Sichuan Province, CSNMT Int. Coop. Res. Centre (MoST), Southwest Jiaotong University, Chengdu 610031, China. Email: xhutang@swjtu.edu.cn}
	}
	
	\maketitle
	
\begin{abstract}
Let $\mathbb{Z}_4$ denote the ring of integers modulo $4$. The Galois ring GR$(4,m)$, which consists of $4^m$ elements, represents the Galois extension of degree $m$ over $\mathbb{Z}_4$. The constructions of codes over $\mathbb{Z}_4$ have garnered significant interest in recent years. In this paper, building upon previous research, we utilize the defining-set approach to construct several classes of linear codes over $\mathbb{Z}_4$ by effectively using the properties of the trace function from GR$(4,m)$ to $\mathbb{Z}_4$. As a result, we have been able to obtain new linear codes over $\mathbb{Z}_4$ with good parameters and determine their Lee weight distributions. Upon comparison with the existing database of $\mathbb{Z}_4$ codes, our construction can yield novel linear codes, as well as linear codes that possess the best known minimum Lee distance.
\end{abstract}
		

\section{Introduction}\label{sec-intro}
	
The exploration of codes over $\zfour$ dates back to the early 1970s. In a groundbreaking paper \cite{Hammons1994}, Hammons et al. constructed a variety of $\zfour$-linear codes, including the Kerdock, Preparata, Goethals, Delsarte-Goethals, and Goethals-Delsarte codes. These constructions highlighted the $\zfour$-linearity of these codes, which offered enhanced error correction capabilities and greater coding efficiency. This seminal work underscored the significance of $\zfour$-linear codes in the realm of cryptography and piqued the interest of researchers in the field. Subsequently, in 1997, Wan delved into the intricate connections between lattices, designs, low correlation sequences, and codes over $\zfour$ in his influential book \cite{Wan1997}, which has significantly deepened and broadened our understanding of $\zfour$ codes and their pivotal role in coding theory. It has also inspired scholars to delve into the study of the constructions and properties of $\zfour$ codes.

Over the past three decades, the construction of codes over $\zfour$ has been a vibrant area of research, with a variety of mathematical tools being employed in these studies. Functions on the Galois ring have been particularly influential in this domain. Schmidt \cite{KUSchmidt2009} constructed several families of constant-amplitude codes over $\zfour$ and explored their connections with bent functions over $\zfour$. Following this, Li et al. \cite{Li2011} investigated an infinite class of codes over $\zfour$ by defining codewords based on $\zfour$-valued quadratic forms, and they determined the Lee and Hamming weight distributions of these codes by analyzing the distribution of values of a class of exponential sums over $\zfour$. More recently, Ban et al. \cite{Ban2022} constructed a family of codes over $\zfour$ derived from generalized bent functions. Using two classes of functions over $\zfour$, Wu et al. \cite{WY.LB2024} constructed an infinite class of $\zfour$-linear codes and an infinite class of nonlinear codes over $\zfour$, and completely determined their Lee weight distributions. Additionally, simplicial complexes and posets have been recognized as effective tools for constructing codes over finite rings. Zhu et al. \cite{Zhu2020} constructed three classes of linear codes over $\zfour$ using posets of the disjoint union of two chains and determined their Lee weight distributions. Wu et al. \cite{Wu2024} constructed some optimal $\zfour$-linear codes via simplicial complexes and determined their Lee weight distributions, which include two infinite classes of 4-weight codes. For further exploration of $\zfour$ code constructions, readers are directed to \cite{Helleseth1996, Pujol2009, Liu2019, Liu2018, Shi2019, Shi2017, Cao2022, Lin2023, Meng2024, Hyun2024}. Recent studies have continued to explore new constructions of codes over $\zfour$, employing various mathematical tools and techniques. However, the literature indicates that there are still only a few known results on the construction of $\zfour$ codes, especially for $\zfour$-linear codes.

Inspired by the works of \cite{Li2011,Shi2019,WY.LB2024}, this paper concentrates on the construction of several classes of linear codes over $\zfour$ utilizing the properties of the trace functions on the Galois ring $\rr$=GR$(4,m)$ through the defining-set approach. Let $\tei = \{ z\in\rr: z^{2^m} = z \}$ be the Teichm\"{u}ller set of $\rr$ and then every element $z\in\rr$ can be uniquely expressed as $z=x+2y$ for $x,y\in\tei$. The trace function from $\rr$ to $\zfour$ is defined as $\Tr_1^m(z)=\sum_{j=0}^{m-1} (x^{2^j} + 2y^{2^j})$ for any $z=x+2y\in\rr$ and $x, y\in\tei$. Let $D=\{d_1,d_2.\cdots,d_n\}\subset\rr$ and $\cc_D$ be defined as follows:
\begin{equation}\label{eq-code}
	\cc_D = \Big\{ \big(\Tr_1^m(ud_1), \Tr_1^m(ud_2), \cdots, \Tr_1^m(ud_n)\big): u \in \rr \Big\}.
\end{equation}
It can be readily verified that $\cc_D$ is a linear code of length $n$ over $\zfour$. In this paper, for $t=0,1,2,3$, define
\begin{equation}\label{eq-Dt}
	D_{t}= \{x\in\tei: \Tr_1^m(x) = t\}.
\end{equation}
We delve into the study of linear codes over $\zfour$ as defined by \eqref{eq-code}, utilizing defining sets $D=D_{t_1}$, $D=D_{t_1}\cup D_{t_2}$ and $D=D_{t_1}\cup D_{t_2}\cup D_{t_3}$, with $t_1,t_2,t_3\in\zfour$. To achieve this, we introduce two exponential sums parameterized by $u\in\rr$, which are intrinsically linked to the trace function mapping from $\rr$ to $\zfour$. Through a meticulous analysis of the value distributions of these two exponential sums and their associated exponential sums, we are able to deduce the parameters and Lee weight distributions of the $\zfour$-linear codes derived from our construction when $m$ is odd. Our methodical approach can yield a variety of $\zfour$-linear codes with few weights. It is significant to highlight that, as per experimental data, our codes not only include some with the best known minimum Lee distance but also introduce novel codes over $\zfour$ when contrasted with the existing database of $\zfour$ codes \cite{Codetable}. Regarding the scenario where $m$ is even, our current findings only allow us to speculate on the possible weights of the aforementioned codes, leaving the weight distribution as an unresolved challenge.

The structure of this paper is outlined as follows. In Section \ref{codes}, we present our main theorems and showcase some novel codes and codes with the best known minimum Lee distance established by our theorems. In Section \ref{sec-auxiliary}, we elaborate on fundamental concepts concerning the ring $\rr$ and its Teichm\"{u}ller set, followed by an exploration of two exponential sums related to the trace function and other supplementary results that will serve as the foundation for proving our theorems. The subsequent three sections are dedicated to the proofs of our three main theorems. Finally, in Section \ref{sec-conclusion}, we summarize the findings of this paper.

\section{Main results}\label{codes}
	
The Lee metric is a crucial measure in the analysis of codes defined over the ring $\zfour$. The Lee weight of an element $a\in\zfour$ is defined as $\omega_L(a)=1-\mathfrak{Re}(i^a)$, where $i = \sqrt{-1}$ and $\mathfrak{Re}(x)$ represents the real part of the complex number $x$. Then the Lee weight of a codeword $\bc = (c_1, c_2, \cdots, c_n)\in\zfour^n$ is defined as $\omega_L(\bc)=n-\mathfrak{Re}(\sum_{j=1}^{n}i^{c_j})$ and the Lee distance of $x,y\in\zfour^n$ is defined as $\omega_L(x-y)$. Codes that differ solely by a permutation of their coordinate positions are referred to as permutation-equivalent. It is shown that any linear code $\mathcal{C}$ over $\zfour$ is permutation-equivalent to a code that possesses a generator matrix of the following form:
\begin{eqnarray}\label{eq-k1k2}
 \left(
	\begin{array}{ccc}
		I_{k_1} & A_1 & B_1+2B_2 \\
		0 & 2I_{k_2} & 2A_2 \\
	\end{array}
	\right),
\end{eqnarray}
where $A_1,A_2,B_1,B_2$ are matrices with entries $0$ or $1$ and $I_k$ represents the identity matrix of size $k$.
If a linear code $\mathcal{C}$ of length $n$ has type $4^{k_1}2^{k_2}$ and minimum Lee distance $d_L$, then it is referred to as an $[n, k_1, k_2, d_L]$ code.
	
Throughout this paper, let $m$ be odd and define	
\begin{eqnarray}\label{eq-tau-sig}	
\sigma=
	\left\{ \begin{array}{cl}
		0, &  {\rm if\;}  3\nmid m, \\
		-3, &  {\rm if\;}  3\mid m,
	\end{array}\right.	
	\;\;
\tau=
	\left\{ \begin{array}{cl}
		i^\frac{m-1}{2}, &   {\rm if}\ m \equiv 1 \, {\rm mod}\, 4, \\
		i^\frac{m+1}{2}, &   {\rm if}\ m \equiv 3 \, {\rm mod}\, 4.
	\end{array}\right.	
\end{eqnarray}

\begin{thm}\label{thm1}
Let $m>3$ be odd, $\cc_D$ and $D_t$ defined as \eqref{eq-code} and \eqref{eq-Dt} respectively, and $D=D_t$. Then for $t=0,1,2,3$, $\cc_D$ has parameters as follows:
\begin{center} 
	\begin{tabular}{|c| c| c| c|}
		\hline $t$  &  $n$ & {\rm Number of codewords} & $d_L$ \\ \hline
		$0$  &  $2^{m-2}+2^\frac{m-3}{2}\tau$ & $4^{m-1}$ & $2^{m-2}+2^\frac{m-3}{2}(\tau-2)$  \\ \hline				
		$1$  &  $2^{m-2}+2^\frac{m-3}{2}i^{m-1}\tau$ & $4^m$ & $2^{m-2}+2^\frac{m-3}{2}(i^{m-1}\tau-2)$ \\ \hline				
		$2$  &  $2^{m-2}-2^\frac{m-3}{2}\tau$ & $2^{2m-1}$ & $2^{m-2}-2^\frac{m-3}{2}(\tau+2)$ \\ \hline				
		$3$  &  $2^{m-2}-2^\frac{m-3}{2}i^{m-1}\tau$ & $4^m$ &  $2^{m-2}-2^\frac{m-3}{2}(i^{m-1}\tau+2)$ \\ \hline
	\end{tabular}
\end{center}
Moreover, the Lee weight distribution of $\cc_D$ for $t=0$ is given by
\begin{center} 
	\begin{tabular}{l l}
		\hline {\rm Lee Weight}  &  {\rm Frequency} \\ \hline
		$0$  &  $1$  \\ \hline
		$2^{m-2}$  &  $2^{m-2}+2^\frac{m-3}{2}\tau-1$ \\ \hline	
		$2^{m-2}+2^{\frac{m-1}{2}}\tau$  &  $2^{m-2}-2^\frac{m-3}{2}\tau$ \\ \hline				
        $2^{m-2}+2^{\frac{m-3}{2}}\tau$  &  $2^{m-4}(3\cdot2^{m-1}-5+3\sigma)-2^\frac{m-3}{2}(2^{m-2}-1)\tau$ \\ \hline
		$2^{m-2}+2^{\frac{m-3}{2}}(\tau\pm1)$  &  $2^{m-3}(2^{m-1}\mp2^\frac{m+1}{2}-1-\sigma\pm2\tau)$ \\ \hline
		$2^{m-2}+2^{\frac{m-3}{2}}(\tau\pm2)$  &  $2^{m-5}(2^{m-1}+1+\sigma\mp4\tau)+2^\frac{m-5}{2}(2^{m-2}-1)(\tau\mp2)$ \\ \hline
	\end{tabular}
\end{center}
the Lee weight distribution of $\cc_D$ for $t=2$ is given by
\begin{center} 
	\begin{tabular}{l l}
		\hline {\rm Lee Weight}  &  {\rm Frequency} \\ \hline
		$0$  &  $1$  \\ \hline
		$2^{m-2}$  &  $2^{m-1}-1$ \\ \hline				
		$2^{m-1} - 2^{\frac{m-1}{2}}\tau$  &  $1$ \\ \hline			
		$2^{m-2} - 2^{\frac{m-1}{2}}\tau$  &  $2^{m-1}-1$ \\ \hline					
		$2^{m-2} - 2^{\frac{m-3}{2}}\tau$  &  $2^{m-3}(3\cdot2^{m-1}-1+3\sigma)+2^\frac{m-1}{2}(2^{m-2}-1)\tau$ \\ \hline			
		$2^{m-2} - 2^{\frac{m-3}{2}}(\tau\pm1)$  &  $2^{m-2}(2^{m-1}-1-\sigma) \ {\rm (each)}$ \\ \hline			
        $2^{m-2} - 2^{\frac{m-3}{2}}(\tau\pm2)$  & $2^{m-4}(2^{m-1}-3+\sigma)-2^\frac{m-7}{2}(2^m-4)\tau \ {\rm (each)}$ \\ \hline
	\end{tabular}
\end{center}
and the Lee weight distributions of $\cc_D$ for $t=1,3$ are given by
\begin{center} 
	\begin{tabular}{l l}
		\hline {\rm Lee Weight}  &  {\rm Frequency} \\ \hline
		$0$  &  $1$  \\ \hline			
		$2^{m-1} - 2^{\frac{m-1}{2}}i^{m+t}\tau$  &  $1$ \\ \hline		
		$2^{m-2}$  &  $2^{m-1}-1$ \\ \hline
		$2^{m-2} - 2^{\frac{m-1}{2}}i^{m+t}\tau$  &  $2^{m-1}-1$ \\ \hline		
		$2^{m-2} - 2^{\frac{m-3}{2}}i^{m+t}\tau$  &  $(2^m-2)\big( 3\cdot2^{m-3} + 2^\frac{m-3}{2}i^{m+t}\tau \big)+2^m $ \\ \hline		
		$2^{m-2} + 2^{\frac{m-3}{2}}i^{t-1}(i^{m-1}\tau\pm1)$  &  $2^{m-2}(2^m-2)\ {\rm (each)}$ \\ \hline			
		$2^{m-2} + 2^{\frac{m-3}{2}}i^{t-1}(i^{m-1}\tau\pm2)$  &  $(2^m-2)\big( 2^{m-4} - 2^{\frac{m-5}{2}}i^{m+t}\tau \big)\ {\rm (each)}$ \\ \hline
	\end{tabular}
\end{center}
\end{thm}
	
\begin{thm}\label{thm2}
Let $m>1$ be odd, $\cc_D$ and $D_t$ defined as \eqref{eq-code} and \eqref{eq-Dt} respectively, and $D=D_{t_1}\cup D_{t_2}$, where $t_1,t_2\in\zfour$. Then, when $t_1$ and $t_2$ have the same parity $\cc_D$ has parameters as follows:

\begin{center} 
	\begin{tabular}{|c| c| c| c|}
		\hline $(t_1,t_2)$  &  $n$ & {\rm Number of codewords} & $d_L$ \\ \hline
		$(0,2)$  &  $2^{m-1}$ & $2^{2m-1}$ & $2^{m-1}-2^\frac{m-1}{2}$  \\ \hline				
		$(1,3)$  &  $2^{m-1}$ & $4^m$ & $2^{m-1}-2^\frac{m-1}{2}$ \\ \hline				
	\end{tabular}
\end{center}
Moreover, the Lee weight distribution of $\cc_D$ for $(t_1,t_2)=(0,2)$ is given by
\begin{center} 
	\begin{tabular}{l l}
		\hline {\rm Lee Weight}  &  {\rm Frequency} \\ \hline
		$0$  & $1$  \\ \hline
				
		$2^{m-1}$  &  $2^{2m-2}-1$ \\ \hline
				
		$2^{m-1}\pm2^{\frac{m-1}{2}}$  &  $2^{2m-3}\mp2^\frac{m-3}{2}(2^m-1)$ \\ \hline	
	\end{tabular}
\end{center}
and the Lee weight distribution of $\cc_D$ for $(t_1,t_2)=(1,3)$ is given by
\begin{center} 
	\begin{tabular}{l l}
		\hline {\rm Lee Weight}  &  {\rm Frequency} \\ \hline
		$0$  &  $1$  \\ \hline			
		$2^m$  &  $1$  \\ \hline			
		$2^{m-1}$  &  $2^{m-1}(2^m+2)-2$ \\ \hline			
		$2^{m-1}\pm2^{\frac{m-1}{2}}$  &  $2^{2m-2}-2^{m-1}\ {\rm (each)}$  \\ \hline		
	\end{tabular}
\end{center}
\end{thm}
	
\begin{rem}\label{rem-t1t2}
When $t_1$ and $t_2$ have different parity, computer experiments suggest that the parameters of the code $\cc_D$ are less favorable than when $t_1$ and $t_2$ share the same parity. As a result, we have not considered the scenario where $t_1$ and $t_2$ have different parity. The Lee weight distribution of the code $\cc_D$ when $t_1$ and $t_2$ have different parity may be of independent interest, which requires further exploration of the properties of trace functions and exponential sums over $\zfour$. 
\end{rem}
	
\begin{thm}\label{thm3}
Let $m>3$ be odd, $\cc_D$ and $D_t$ defined as \eqref{eq-code} and \eqref{eq-Dt} respectively, and $D=D_{t_1}\cup D_{t_2}\cup D_{t_3}=\tei\backslash D_t$, where $t_1,t_2,t_3,t\in\zfour$. Then, $\cc_D$ has parameters as follows:

\begin{center} 
	\begin{tabular}{|c| c| c| c| c|}
		\hline $t$ & $(t_1,t_2,t_3)$ & $n$ & {\rm Number of codewords} & $d_L$ \\ \hline
		$0$ & $(1,2,3)$ & $3\cdot2^{m-2}-2^\frac{m-3}{2}\tau$ & $4^m$ & $3\cdot2^{m-2}-2^\frac{m-3}{2}(\tau+3)$ \\ \hline				
		$1$ & $(0,2,3)$ & $3\cdot2^{m-2}-2^\frac{m-3}{2}i^{m-1}\tau$ & $4^m$ & $2^{m-1}-2^\frac{m-1}{2}i^{m-1}\tau$ \\ \hline
		$2$ & $(0,1,3)$ & $3\cdot2^{m-2}+2^\frac{m-3}{2}\tau$ & $4^m$ & $2^{m-1}$ \\ \hline	
		$3$ & $(0,1,2)$ & $3\cdot2^{m-2}+2^\frac{m-3}{2}i^{m-1}\tau$ & $4^m$ & $2^{m-1}+2^\frac{m-1}{2}i^{m-1}\tau$ \\ \hline			
	\end{tabular}
\end{center}
Moreover, the Lee weight distribution of $\cc_D$ for $t=0$ is given by
\begin{center} 
	\begin{tabular}{l l}
		\hline {\rm Lee Weight}  &  {\rm Frequency} \\ \hline
		$0$  &  $1$  \\ \hline			
		$2^m$  &  $1$ \\ \hline			
		$2^m - 2^\frac{m-1}{2}\tau$  &  $2$ \\ \hline			
		$3\cdot2^{m-2}$  &  $2^m-2$ \\ \hline			
		$3\cdot2^{m-2}-2^\frac{m-1}{2}\tau$  &  $2^m-2$ \\ \hline
		$3\cdot2^{m-2} - 2^\frac{m-3}{2}\tau$  &  $2^{m-2}(2^{m-1}+1+\sigma)+2^\frac{m-3}{2}(2^m-4)\tau$ \\ \hline			
		$3\cdot2^{m-2}-2^\frac{m-3}{2}(\tau\pm1)$  &  $3\cdot2^{m-3}\big(2^{m-1}\pm2^\frac{m+1}{2}-1-\sigma\mp2\tau \big)$ \\ \hline			
        $3\cdot2^{m-2}-2^\frac{m-3}{2}(\tau\pm2)$  &  $2^{m-3}(3\cdot2^{m-1}-5+3\sigma)-2^\frac{m-5}{2}(2^m-4)\tau \ {\rm (each)}$ \\ \hline
        $3\cdot2^{m-2}-2^\frac{m-3}{2}(\tau\pm3)$  &  $2^{m-3}(2^{m-1}\mp2^\frac{m+1}{2}-1-\sigma\pm2\tau)$ \\ \hline
	\end{tabular}
\end{center}
the Lee weight distribution of $\cc_D$ for $t=2$ is given by
\begin{center} 
	\begin{tabular}{l l}
		\hline {\rm Lee Weight}  &  {\rm Frequency} \\ \hline
		$0$  &  $1$  \\ \hline			
		$2^m$  &  $1$ \\ \hline			
		$2^{m-1}$  &  $2$ \\ \hline			
		$3\cdot2^{m-2}$  &  $2^m+2^\frac{m+1}{2}\tau-4$ \\ \hline
		$3\cdot2^{m-2} + 2^\frac{m-1}{2}\tau$  &  $2^m-2^\frac{m+1}{2}\tau$ \\ \hline			
		$3\cdot2^{m-2} + 2^\frac{m-3}{2}\tau$  &  $2^{m-2}(2^{m-1}-3+\sigma)-2^\frac{m-3}{2}(2^m-4)\tau$ \\ \hline
	    $3\cdot2^{m-2}+2^\frac{m-3}{2}(\tau\pm1)$  &  $3\cdot2^{m-3}(2^{m-1}-1-\sigma)\mp2^{m-2}(2^\frac{m-1}{2}-\tau)$ \\ \hline					
		$3\cdot2^{m-2}+2^\frac{m-3}{2}(\tau\pm2)$  &  $2^{m-3}\big( 3\cdot2^{m-1}-1+3\sigma\mp4\tau \big)+2^\frac{m-5}{2}(2^m-4)(\tau\mp2)$ \\ \hline						
		$3\cdot2^{m-2}+2^\frac{m-3}{2}(\tau\pm3)$  &  $2^{m-3}\big(2^{m-1}\mp2^\frac{m+1}{2}-1-\sigma\pm2\tau\big)$ \\ \hline
	\end{tabular}
\end{center}
and the Lee weight distributions of $\cc_D$ for $t=1,3$ are given by
\begin{center} 
	\begin{tabular}{l l}
		\hline {\rm Lee Weight}  &  {\rm Frequency} \\ \hline
		$0$  &  $1$  \\ \hline			
		$2^{m-1} + 2^\frac{m-1}{2}i^{m+t}\tau$  &  $1$ \\ \hline	
		$3\cdot2^{m-2}$  &  $2^{m-1}-1$ \\ \hline
		$3\cdot2^{m-2} + 2^\frac{m-1}{2}i^{m+t}\tau$  &  $2^{m-1}-1$ \\ \hline				
		$3\cdot2^{m-2} - 2^\frac{m-3}{2}i^{m+t}\tau$  &  $2^{m-1} + 2^\frac{m-1}{2}i^{m+t}\tau$ \\ \hline
		$3\cdot\big(2^{m-2} + 2^\frac{m-3}{2}i^{m+t}\tau\big)$  &  $2^{m-1} - 2^\frac{m-1}{2}i^{m+t}\tau$ \\ \hline		
        $3\cdot2^{m-2} + 2^\frac{m-3}{2}i^{m+t}\tau$  &  $(2^m-2)\big( 2^{m-3} - 2^\frac{m-3}{2}i^{m+t}\tau \big)$ \\ \hline			
		$3\cdot2^{m-2}-2^\frac{m-3}{2}i^{t-1}(i^{m-1}\tau\pm 1)$  &  $(2^m-2)\big( 3\cdot2^{m-4} \pm 2^\frac{m-5}{2}i^{t-1} \big)$ \\ \hline
		$3\cdot2^{m-2}-2^\frac{m-3}{2}i^{t-1}(i^{m-1}\tau\pm 2)$  &  $(2^m-2)\big(3\cdot2^{m-4} - 2^\frac{m-5}{2}i^{t-1}(i^{m-1}\tau\mp2)\big)$ \\ \hline			
		$3\cdot2^{m-2}-2^\frac{m-3}{2}i^{t-1}(i^{m-1}\tau\pm 3)$  &  $(2^m-2)\big( 2^{m-4} \pm 2^\frac{m-5}{2}i^{t-1} \big)$ \\ \hline
	\end{tabular}
\end{center}
\end{thm}
	
\begin{rem}\label{rem-goodcode}
Computer experiments show that many good linear codes over $\zfour$ can be obtained from our constructions. Table \ref{tab-newcode} displays a selection of $\zfour$-linear codes derived from our theorems, in which ``Best known $d_L$" denotes the best known minimum Lee distance of $\zfour$-linear codes with the parameters $[n,k_1,k_2]$ based on the current database of $\zfour$ codes. The symbol ``-" indicates the absence of a linear code with the same length and type in the existing database.
\end{rem}
	
\begin{table} [ht]
	\caption{Some $\zfour$-linear codes from our theorems for small values of $m$} \label{tab-newcode}
	\begin{center} 
	   \begin{tabular}{c c c c c}
			\hline $m$  &  $[n, k_1, k_2]$  &  {\rm Best known} $d_L$  & {\rm Our} $d_L$  &  {\rm Theorem} \\ \hline		
			$3$  &  $[4, 2, 1]$  &  -  &  $2$  &  Theorem \ref{thm2} \\ \hline
			$3$  &  $[4, 3, 0]$  &  $2$  &  $2$  &  Theorem \ref{thm2} \\ \hline
			$5$  &  $[6, 5, 0]$  &  $2$  &  $2$  &  Theorem \ref{thm1} \\ \hline
			$5$  &  $[10, 4, 1]$  &  -  &  $6$  &  Theorem \ref{thm1} \\ \hline
			$5$  &  $[10, 5, 0]$  &  $6$  &  $6$  &  Theorem \ref{thm1} \\ \hline
			$5$  &  $[16, 4, 1]$  &  $8$  &  $12$  &  Theorem \ref{thm2} \\ \hline
			$5$  &  $[16, 5, 0]$  &  $10$  &  $12$  &  Theorem \ref{thm2} \\ \hline
			$5$  &  $[22, 5, 0]$  &  $14$  &  $16$  &  Theorem \ref{thm3} \\ \hline
			$5$  &  $[26, 5, 0]$  &  $16$  &  $20$  &  Theorem \ref{thm3} \\ \hline
			$7$  &  $[28, 6, 1]$  &  $12$  &  $20$  &  Theorem \ref{thm1} \\ \hline
			$7$  &  $[28, 7, 0]$  &  $13$  &  $20$  &  Theorem \ref{thm1} \\ \hline
			$7$  &  $[36, 6, 0]$  &  $28$  &  $28$  &  Theorem \ref{thm1} \\ \hline
			$7$  &  $[36, 7, 0]$  &  $28$  &  $28$  &  Theorem \ref{thm1} \\ \hline
			$7$  &  $[64, 6, 1]$  &  $53$  &  $56$  &  Theorem \ref{thm2} \\ \hline
			$7$  &  $[64, 7, 0]$  &  $33$  &  $56$  &  Theorem \ref{thm2} \\ \hline
			$7$  &  $[92, 7, 0]$  &  $78$  &  $80$  &  Theorem \ref{thm3} \\ \hline		
			$7$  &  $[100, 7, 0]$  &  $50$  &  $64$  &  Theorem \ref{thm3} \\ \hline
			$7$  &  $[100, 7, 0]$  &  $50$  &  $72$  &  Theorem \ref{thm3} \\ \hline	
			$9$  &  $[120, 8, 1]$  &  $96$  &  $104$  &  Theorem \ref{thm1} \\ \hline
			$9$  &  $[120, 9, 0]$  &  $96$  &  $104$  &  Theorem \ref{thm1} \\ \hline	
		\end{tabular}
	\end{center}
\end{table}

\begin{rem}\label{rem-even}
In this paper, we focus on the case where $m$ is an odd integer, for which we determine the weight distributions of the constructed codes. Nevertheless, the type of these codes, specifically, determining the precise values of $k_1$ and $k_2$ as defined in \eqref{eq-k1k2}, remains unresolved. When $m$ is even, the linear codes generated by our constructions do not perform as well as those with odd $m$. Specifically, for even $m$, these codes exhibit a larger number of distinct Lee weights, and fewer of them achieve the best known minimum Lee distance according to the existing database of $\zfour$ codes. Our method allows us to identify the possible Lee weights for codes with even $m$. However, we are currently unable to determine the Lee weight distributions of these codes, which remains an open problem for further research.
\end{rem}

\section{Auxiliary lemmas}\label{sec-auxiliary}
	
It should be noted that the Teichm\"{u}ller set $\tei$ of $\rr$ is not closed under the addition operation. Nonetheless, it is significant to observe that any element $z\in\tei$ can be uniquely represented as $z = x+y+2\sqrt{xy}$ for some $x, y \in \tei$. For convenience, define an addition operation $\oplus$ in the Teichm\"{u}ller set $\tei$ as
$$ x \oplus y = x + y + 2\sqrt{xy} $$
for all $x, y \in \tei$. Then it can be readily verified that $(\tei, \oplus, \cdot)$ is isomorphic to $(\f,+,\cdot)$.
Similar to the trace function over $\f$, define the trace function from $(\tei, \oplus, \cdot)$ to $\ztwo$ as
\begin{equation*}
	\tr_1^m(x) = \oplus_{j=0}^{m-1} x^{2^j},\ x\in\tei.
\end{equation*}

The relation between the trace functions $\Tr_1^m(x)$ and $\tr_1^m(x)$ is described as follows:
\begin{lem}{\rm(\cite{Hammons1994})}\label{lem-Tr-tr}		
	The trace function over GR$(4,m)$ has $2$-adic expansion given by
	\[\Tr_1^m(x) = \tr_1^m(x)+2p(x), \;\; x\in\tei,\]
	where $p(x)$ is defined as	
	\begin{eqnarray*}
		p(x)=\left\{\begin{array}{ll}
		\oplus_{j=1}^{(m-1)/2}\tr_1^m(x^{2^j+1}), & {\rm if}\;m\; {\rm is}\;  {\rm odd}, \\
		\oplus_{j=1}^{m/2-1}\tr_1^m(x^{2^j+1})\oplus\tr_1^{m/2}(x^{2^{m/2}+1}), & {\rm if}\;m\; {\rm is}\;  {\rm even}.
	   \end{array}\right.
	\end{eqnarray*}
\end{lem}
Leveraging the properties of the trace function over finite fields, Lemma \ref{lem-Tr-tr} readily demonstrates that $2\Tr_1^m(x) = 2\tr_1^m(x)$ and $\Tr_1^m(x^{2^k}) = \Tr_1^m(x)$ hold for any positive integer $k$ and $x\in\tei$, which will be frequently utilized in our subsequent proofs.

Let $u = a + 2b\in\rr$ where $a, b\in\tei$. Define $Q_a(x) = \Tr_1^m(ax)$ and
$$ \chi_{Q_a}(b) = \sum_{x\in\tei}i^{Q_a(x)}(-1)^{\tr_1^m(bx)} = \sum_{x\in\tei}i^{\Tr_1^m(ux)}. $$
	
The following lemmas about the values of $\chi_{Q_a}(b)$ will be used in the subsequent proofs.
	
\begin{lem}{\rm(\cite{Schmidt2009})}\label{lem-chi-non}
    For any $a\ne 0$, the value distribution of the multiset $\{ \chi_{Q_a}(b):\ b\in\tei \}$ for odd $m$ is given as follows:
	$$ \left\{\begin{array}{cl}
		\pm 2^\frac{m-1}{2}(1+i), & 2^{m-2} \pm 2^\frac{m-3}{2}\ {\rm times}, \\
		\pm 2^\frac{m-1}{2}(1-i), & 2^{m-2} \pm 2^\frac{m-3}{2}\ {\rm times}.
	\end{array}\right. $$	
\end{lem}

The following result is a direct consequence of Lemma 3 and Lemma 4 in \cite{Yang1996}.

\begin{lem}{\rm(\cite{Yang1996})}\label{lem-chi1-0}
	For odd $m$ and any $u = a + 2b \in\rr$ with $a,b\in\tei$ and $a\ne 0$, we have
	$$ \chi_{Q_a}(b) = i^{-\Tr_1^m(\frac{b}{a})} \chi_{Q_1}(0), $$ where
	$\chi_{Q_1}(0) = 2^\frac{m-1}{2}\tau(1+i^m)$, and the values of $\tau$ are given in \eqref{eq-tau-sig}.	
\end{lem}
	
In order to determine the parameters and the weight distributions of $\zfour$-linear codes as presented in our theorems, we introduce two classes of exponential sums defined by
\begin{equation}\label{eq-S}
	S_+(u) = \sum_{x\in\tei}i^{\Tr_1^m(ux)}+\sum_{x\in\tei}i^{3\Tr_1^m(ux)},\
	S_-(u) = \sum_{x\in\tei}i^{\Tr_1^m(ux)}-\sum_{x\in\tei}i^{3\Tr_1^m(ux)},
\end{equation}
where $\tei$ denotes the Teichm\"{u}ller set. We then proceed to analyze the distribution of the values of these two exponential sums. For ease of reference, we define $\tei^*=\tei \backslash\{0\}$ and $\tei^\star=\tei \backslash\{0,1\}$, and these notations will be used throughout the paper.

\begin{lem}\label{lem-S1-dis}
    Let $m$ be an odd integer and $u = a+2b\in\rr$. When $a,b$ run through $\tei$, the value distributions of $S_+(u)$ and $S_-(u)$ are given by
	$$S_+(u) = \left\{\begin{array}{cl}
		2^{m+1}, & \ {\rm once}, \\
		0, & 2^m-1\ {\rm times}, \\
		\pm2^{\frac{m+1}{2}}, & (2^m-1)\big(2^{m-1} \pm 2^{\frac{m-1}{2}}\big)\ {\rm times.}
	\end{array}\right.$$
	$$S_-(u) = \left\{\begin{array}{cl}
		0, & 2^m\ {\rm times}, \\
		\pm2^{\frac{m+1}{2}}i, & 2^{m-1}(2^m-1)\ {\rm times\ (each).}
	\end{array}\right.$$
\end{lem}
\begin{proof}
Observe by \eqref{eq-S} that
$$ S_+(u) = 2\re\big( \sum_{x\in\tei}i^{\Tr_1^m(ux)} \big) = 2\re\big( \sum_{x\in\tei}i^{\Tr_1^m(ax)}(-1)^{\tr_1^m(bx)} \big). $$
If $a = 0$, then $S_+(u) = 2\re\big(\sum_{x\in\tei}(-1)^{\tr_1^m(bx)}\big)=2^{m+1}$ if $b=0$, and $0$ otherwise.
If $a\in\tei^*$, then by $S_+(u) = 2\re(\chi_{Q_a}(b))$ and Lemma \ref{lem-chi-non}, one can conclude that in this case the values of $S_+(u)=\pm2^{\frac{m+1}{2}}$ occur $(2^m-1)\big(2^{m-1}\pm 2^{\frac{m-1}{2}}\big)$ times, respectively. Thus, the first assertion is established by combining the above two cases. Moreover, the second one can be similarly obtained by the fact that $S_-(u)=2\im\big(\chi_{Q_a}(b)\big)i$, where $\im(x)$ represents the imaginary part of the complex number $x$. This completes the proof.
\end{proof}

To ascertain the weight distributions of our codes as outlined in Theorems \ref{thm1}-\ref{thm3}, additional identities involving $S_+(u)$ and $S_-(u)$ are required.

\begin{lem}\label{lem-S1}
	Let $m$ be odd and $u=a+2b\in\rr$ with $a\in\tei^\star$ and $b\in\tei$. Then
	\begin{enumerate}
		\item[\rm(1)] $ \sum\limits_{a\in \tei^\star}\sum\limits_{b\in \tei} S_+(u) = 2^{m+1}(2^m-2)$;
		\item[\rm(2)] $ \sum\limits_{a\in \tei^\star}\sum\limits_{b\in \tei} S_-(u) = 0$.
	\end{enumerate} 	
\end{lem}
\begin{proof}
By \eqref{eq-S} and the fact $\sum_{b\in \tei}(-1)^{\tr_1^m(bx)}=2^m$ if $x=0$ and 0 otherwise, one gets
	\begin{eqnarray*}
		\sum_{a\in \tei^\star}\sum_{b\in \tei} S_+(u) & = &
		\sum_{a\in \tei^\star}\sum_{b\in \tei}\big( \sum_{x\in\tei}i^{\Tr_1^m(ax+2bx)} + \sum_{x\in\tei}i^{\Tr_1^m(3ax+2bx)} \big) \\
		& = & \sum_{x\in\tei}\sum_{a\in \tei^\star}i^{\Tr_1^m(ax)} \sum_{b\in \tei}(-1)^{\tr_1^m(bx)}+\sum_{x\in\tei}\sum_{a\in \tei^\star}i^{\Tr_1^m(3ax)} \sum_{b\in \tei}(-1)^{\tr_1^m(bx)} \\
		& = & 2^{m+1}(2^m-2).
	\end{eqnarray*}
	The second assertion can be similarly obtained. This completes the proof.
\end{proof}
	
\begin{lem}\label{lem-S12}
	With the same notations, we have
	\begin{enumerate}
		\item[\rm(1)] $ \sum\limits_{a\in \tei^\star}\sum\limits_{b\in \tei}S_+(u)S_+(u+2) = \sum\limits_{a\in \tei^\star}\sum\limits_{b\in \tei}S_-(u)S_-(u+2) = 0$;
		\item[\rm(2)] $ \sum\limits_{a\in \tei^\star}\sum\limits_{b\in \tei}S_+(u)S_+(u+1) = 2^{m+1}\big( 2^m+(2^m-4)2^\frac{m-1}{2}\tau \big) $;
		\item[\rm(3)] $ \sum\limits_{a\in \tei^\star}\sum\limits_{b\in \tei} S_+(u)S_-(u+1) = -\sum\limits_{a\in \tei^\star}\sum\limits_{b\in \tei}S_+(u+1)S_-(u) = 2^\frac{3m+1}{2}(2^m-2)i^m\tau$.
	\end{enumerate}
\end{lem}
\begin{proof}
	We only give the proof for Case (1) here since the others can be proved in the same manner.	Substituting $u=a+2b$, a straightforward calculation yields:
	{\footnotesize
	   $$\begin{aligned}
			& \sum_{a\in \tei^\star}\sum_{b\in \tei} S_+(u)S_+(u+2) \\ =
			& \sum_{a\in \tei^\star}\sum_{b\in \tei} \big(\sum_{x\in\tei}i^{\Tr_1^m(ax+2bx)} + \sum_{x\in\tei}i^{\Tr_1^m(3ax+2bx)}\big)  \big( \sum_{y\in\tei}i^{\Tr_1^m(2y+ay+2by)} + \sum_{y\in\tei}i^{\Tr_1^m(2y+3ay+2by)} \big) \\ =
			& \big(\sum_{a\in \tei^\star}\sum_{x,y\in\tei}(i^{\Tr_1^m(2y+ax+ay)}+i^{\Tr_1^m(2y+ax+3ay)}+i^{\Tr_1^m(2y+3ax+ay)}+ i^{\Tr_1^m(2y+3ax+3ay)}) \big)\sum_{b\in \tei} (-1)^{\tr_1^m(b(x \oplus y))}
			\\ =& 2^m\sum_{a\in \tei^\star}\sum_{y\in\tei} \big( 2i^{\Tr_1^m(2y+2ay)}+ 2i^{\Tr_1^m(2y)} \big)
			\\ =& 0
		\end{aligned}$$
	}
	due to $a\ne 1$.
	Similarly, one can obtain that
	\[\sum_{a\in \tei^\star}\sum_{b\in \tei} S_-(u)S_-(u+2) = 2^m\sum_{a\in \tei^\star}\sum_{y\in\tei} \big( 2i^{\Tr_1^m(2y+2ay)} - 2i^{\Tr_1^m(2y)} \big) = 0. \]
    This completes the proof.
\end{proof}

\begin{lem}\label{lem-S123}
	Let $m$ be odd and $u=a+2b\in\rr$ with $a\in\tei^\star$ and $b\in\tei$. Then
    \begin{enumerate}
	   \item[\rm(1)] $ \sum\limits_{a\in\tei^\star}\sum\limits_{b\in \tei} S_+(u)S_+(u+2)S_+(u+1) = 2^{2m+3}(2^\frac{m-1}{2}\tau-1)$;
	   \item[\rm(2)] $ \sum\limits_{a\in \tei^\star}\sum\limits_{b\in \tei} S_+(u)S_+(u+2)S_-(u+1) = \sum\limits_{a\in \tei^\star}\sum\limits_{b\in \tei} S_+(u)S_-(u+1)S_-(u+3) = 0$.
    \end{enumerate}
\end{lem}
\begin{proof}
	See Appendix \ref{append-a}.
\end{proof}

\begin{lem}\label{lem-S1234}
	Let $m$ be an odd integer, $u=a+2b\in\rr$ with $a\in\tei^\star$ and $b\in\tei$,  and $\sigma$ be given as in \eqref{eq-tau-sig}. Then we have
    \begin{enumerate}
	   \item[\rm(1)] $ \sum\limits_{a\in\tei^\star}\!\sum\limits_{b\in \tei}\! S_+(u)S_+(u+2)S_+(u+1)S_+(u+3) = 2^{3m+3}\sigma$;
	   \item[\rm(2)] $ \sum\limits_{a\in \tei^\star}\sum\limits_{b\in \tei} S_+(u)S_+(u+2)S_-(u+1)S_-(u+3) = 0$.
    \end{enumerate}
\end{lem}
\begin{proof}
    See Appendix \ref{append-b}.
\end{proof}
	
Notice that $S_+(u)=\pm 2^\frac{m+1}{2}$ for $u=a+2b$ when $a\ne 0$ according to Lemma \ref{lem-S1-dis}.

\begin{lem}\label{lem-S1S2}
	Let $m$ be odd and $u=a+2b\in\rr$ with $a\in\tei^\star$ and $b\in\tei$. Then the value distribution of $\big( S_+(u), S_+(u+2) \big)$ is shown as
	\begin{center} 
		\begin{tabular}{c c c}
			\hline $S_+(u)$  &  $S_+(u+2)$  & {\rm Frequency} \\ \hline
			$2^\frac{m+1}{2}$  &  $2^\frac{m+1}{2}$  &  $(2^m-2)(2^{m-2}+2^\frac{m-1}{2})$ \\ \hline
			$2^\frac{m+1}{2}$  &  $-2^\frac{m+1}{2}$  &  $2^{2m-2}-2^{m-1}$ \\ \hline
			$-2^\frac{m+1}{2}$  &  $2^\frac{m+1}{2}$  &  $2^{2m-2}-2^{m-1}$ \\ \hline
			$-2^\frac{m+1}{2}$  &  $-2^\frac{m+1}{2}$  &  $(2^m-2)(2^{m-2}-2^\frac{m-1}{2})$ \\ \hline
		\end{tabular}
	\end{center}
\end{lem}
\begin{proof}
	When $a$ ranges over $\tei^\star$ and $b$ runs through $\tei$, according to Lemma \ref{lem-S1-dis}, assume that the pair $\big( S_+(u), S_+(u+2) \big)$ takes values of $(2^\frac{m+1}{2},2^\frac{m+1}{2})$, $(2^\frac{m+1}{2},-2^\frac{m+1}{2})$, $(-2^\frac{m+1}{2},2^\frac{m+1}{2})$, $(-2^\frac{m+1}{2},-2^\frac{m+1}{2})$ exactly $x_1,x_2,x_3,x_4$ times respectively. Then, by Lemma \ref{lem-S1} (1) and Lemma \ref{lem-S12} (1), one gets
	\begin{equation}\label{eq-xj4}
		\left\{\begin{array}{cl}
			x_1+x_2-x_3-x_4 & = 2^\frac{m+1}{2}(2^m-2), \\
			x_1-x_2+x_3-x_4 & = 2^\frac{m+1}{2}(2^m-2), \\
			x_1-x_2-x_3+x_4 & = 0.
		\end{array}\right.
	\end{equation}
	This together with $x_1+x_2+x_3+x_4 = 2^m(2^m-2)$ completes the proof.
\end{proof}

\begin{table} [ht]
	\caption{Value Distribution of $\big(S_+(u), S_+(u+1), S_+(u+2), S_+(u+3)\big)$} \label{tab-Si}
	\begin{center}
		\begin{tabular}{c c c c c}
			\hline $S_+(u)$  &  $S_+(u+1)$  &  $S_+(u+2)$  &  $S_+(u+3)$  &  {\rm Frequency} \\ \hline
			$2^\frac{m+1}{2}$  &  $2^\frac{m+1}{2}$  &  $2^\frac{m+1}{2}$  &  $2^\frac{m+1}{2}$  &  $x_1$ \\ \hline
			$2^\frac{m+1}{2}$  &  $2^\frac{m+1}{2}$  &  $2^\frac{m+1}{2}$  &  $-2^\frac{m+1}{2}$  &  $x_2$ \\ \hline
			$2^\frac{m+1}{2}$  &  $2^\frac{m+1}{2}$  &  $-2^\frac{m+1}{2}$  &  $2^\frac{m+1}{2}$  &  $x_3$ \\ \hline
			$2^\frac{m+1}{2}$  &  $2^\frac{m+1}{2}$  &  $-2^\frac{m+1}{2}$  &  $-2^\frac{m+1}{2}$  &  $x_4$ \\ \hline
			$2^\frac{m+1}{2}$  &  $-2^\frac{m+1}{2}$  &  $2^\frac{m+1}{2}$  &  $2^\frac{m+1}{2}$  &  $x_5$ \\ \hline
			$2^\frac{m+1}{2}$  &  $-2^\frac{m+1}{2}$  &  $2^\frac{m+1}{2}$  &  $-2^\frac{m+1}{2}$  &  $x_6$ \\ \hline
			$2^\frac{m+1}{2}$  &  $-2^\frac{m+1}{2}$  &  $-2^\frac{m+1}{2}$  &  $2^\frac{m+1}{2}$  &  $x_7$ \\ \hline
			$2^\frac{m+1}{2}$  &  $-2^\frac{m+1}{2}$  &  $-2^\frac{m+1}{2}$  &  $-2^\frac{m+1}{2}$  &  $x_8$ \\ \hline
			$-2^\frac{m+1}{2}$  &  $2^\frac{m+1}{2}$  &  $2^\frac{m+1}{2}$  &  $2^\frac{m+1}{2}$  &  $x_9$ \\ \hline
			$-2^\frac{m+1}{2}$  &  $2^\frac{m+1}{2}$  &  $2^\frac{m+1}{2}$  &  $-2^\frac{m+1}{2}$  &  $x_{10}$ \\ \hline
			$-2^\frac{m+1}{2}$  &  $2^\frac{m+1}{2}$  &  $-2^\frac{m+1}{2}$  &  $2^\frac{m+1}{2}$  &  $x_{11}$ \\ \hline
			$-2^\frac{m+1}{2}$  &  $2^\frac{m+1}{2}$  &  $-2^\frac{m+1}{2}$  &  $-2^\frac{m+1}{2}$  &  $x_{12}$ \\ \hline
			$-2^\frac{m+1}{2}$  &  $-2^\frac{m+1}{2}$  &  $2^\frac{m+1}{2}$  &  $2^\frac{m+1}{2}$  &  $x_{13}$ \\ \hline
			$-2^\frac{m+1}{2}$  &  $-2^\frac{m+1}{2}$  &  $2^\frac{m+1}{2}$  &  $-2^\frac{m+1}{2}$  &  $x_{14}$ \\ \hline
			$-2^\frac{m+1}{2}$  &  $-2^\frac{m+1}{2}$  &  $-2^\frac{m+1}{2}$  &  $2^\frac{m+1}{2}$  &  $x_{15}$ \\ \hline
			$-2^\frac{m+1}{2}$  &  $-2^\frac{m+1}{2}$  &  $-2^\frac{m+1}{2}$  &  $-2^\frac{m+1}{2}$  &  $x_{16}$ \\ \hline
		\end{tabular}
	\end{center}
\end{table}
	
\begin{lem}\label{lem-S1S2S3S4}
	Let $m$ be odd and $u=a+2b\in\rr$ with $a\in\tei^\star$ and $b\in\tei$. Then the value distribution of $\big(S_+(u), S_+(u+1), S_+(u+2), S_+(u+3)\big)$ is given in Table \ref{tab-Si}, where
	\begin{eqnarray*}
		x_1 & = & 2^{m-3}(2^{m-1}+1+\sigma+4\tau)+2^\frac{m-1}{2}(2^{m-2}-1)(2+\tau), \\
		x_2 & = & x_3 = x_5 = x_9 = 2^{m-3}(2^{m-1}+2^\frac{m+1}{2}-1-\sigma-2\tau), \\
		x_4 & = & x_7 = x_{10} = x_{13} = 2^{m-3}(2^{m-1}-1+\sigma), \\
		x_6 & = & x_{11} = 2^{m-3}(2^{m-1}-3+\sigma-2^\frac{m+1}{2}\tau)+2^\frac{m-1}{2}\tau, \\
		x_8 & = & x_{12} = x_{14} = x_{15} = 2^{m-3}(2^{m-1}-2^\frac{m+1}{2}-1-\sigma+2\tau), \\
		x_{16} & = & 2^{m-3}(2^{m-1}+1+\sigma-4\tau)-2^\frac{m-1}{2}(2^{m-2}-1)(2-\tau).
	\end{eqnarray*}
\end{lem}
\begin{proof}
	When $a$ ranges over $\tei^\star$ and $b$ runs through $\tei$, according to Lemma \ref{lem-S1-dis}, assume that the value distribution of $\big(S_+(u), S_+(u+1), S_+(u+2), S_+(u+3)\big)$ is given in Table \ref{tab-Si}. Next we determine the values of $x_l$ for $l=1,2,\cdots,16$. Similar to \eqref{eq-xj4}, using Lemma \ref{lem-S1} (1), Lemma \ref{lem-S12} (1)(2), Lemma \ref{lem-S123} (1) and Lemma \ref{lem-S1234} (1), one can obtain
	$$
	{\footnotesize
		\setcounter{MaxMatrixCols}{24}
		  \left(\begin{matrix}
		        1 & 1 & 1 & 1 & 1 & 1 & 1 & 1 & -1 & -1 & -1 & -1 & -1 & -1 & -1 & -1 \\
				1 & 1 & 1 & 1 & -1 & -1 & -1 & -1 & 1 & 1 & 1 & 1 & -1 & -1 & -1 & -1 \\
				1 & 1 & -1 & -1 & 1 & 1 & -1 & -1 & 1 & 1 & -1 & -1 & 1 & 1 & -1 & -1 \\
				1 & -1 & 1 & -1 & 1 & -1 & 1 & -1 & 1 & -1 & 1 & -1 & 1 & -1 & 1 & -1 \\
				1 & 1 & 1 & 1 & -1 & -1 & -1 & -1 & -1 & -1 & -1 & -1 & 1 & 1 & 1 & 1 \\
				1 & 1 & -1 & -1 & 1 & 1 & -1 & -1 & -1 & -1 & 1 & 1 & -1 & -1 & 1 & 1 \\
				1 & -1 & 1 & -1 & 1 & -1 & 1 & -1 & -1 & 1 & -1 & 1 & -1 & 1 & -1 & 1 \\
				1 & 1 & -1 & -1 & -1 & -1 & 1 & 1 & 1 & 1 & -1 & -1 & -1 & -1 & 1 & 1 \\
				1 & -1 & 1 & -1 & -1 & 1 & -1 & 1 & 1 & -1 & 1 & -1 & -1 & 1 & -1 & 1 \\
				1 & -1 & -1 & 1 & 1 & -1 & -1 & 1 & 1 & -1 & -1 & 1 & 1 & -1 & -1 & 1 \\
				1 & 1 & -1 & -1 & -1 & -1 & 1 & 1 & -1 & -1 & 1 & 1 & 1 & 1 & -1 & -1 \\
				1 & -1 & 1 & -1 & -1 & 1 & -1 & 1 & -1 & 1 & -1 & 1 & 1 & -1 & 1 & -1 \\
				1 & -1 & -1 & 1 & 1 & -1 & -1 & 1 & -1 & 1 & 1 & -1 & -1 & 1 & 1 & -1 \\
				1 & -1 & -1 & 1 & -1 & 1 & 1 & -1 & 1 & -1 & -1 & 1 & -1 & 1 & 1 & -1 \\
				1 & -1 & -1 & 1 & -1 & 1 & 1 & -1 & -1 & 1 & 1 & -1 & 1 & -1 & -1 & 1 \\
				1 & 1 & 1 & 1 & 1 & 1 & 1 & 1 & 1 & 1 & 1 & 1 & 1 & 1 & 1 & 1
		  \end{matrix}\right)
		  \left(\begin{matrix}
				x_1\\ x_2\\ x_3\\ x_4\\ x_5\\ x_6\\ x_7\\ x_8\\ x_9\\ x_{10}\\ x_{11}\\ x_{12}\\
				x_{13}\\ x_{14}\\ x_{15}\\ x_{16}
			\end{matrix}\right)
			=\left(\begin{matrix}
				s_1\\ s_2\\ s_3\\ s_4\\ s_5\\ s_6\\ s_7\\ s_8\\ s_9\\ s_{10}\\ s_{11}\\ s_{12}\\
				s_{13}\\ s_{14}\\ s_{15}\\ s_{16}
			\end{matrix}\right)
		}
		$$
		where $s_{16}= 2^m(2^m-2)$ and
		\begin{eqnarray*}
			s_1 & = & s_2 = s_3 = s_4 = 2^{-\frac{(m+1)}{2}}\sum\nolimits_{a\in\tei^\star}\sum\nolimits_{b\in\tei}S_+(u), \\
			s_5 & = & s_7 = s_8 = s_{10} = 2^{-(m+1)}\sum\nolimits_{a\in\tei^\star}\sum\nolimits_{b\in\tei}S_+(u)S_+(u+1), \\
			s_6 & = & s_9 = 2^{-(m+1)}\sum\nolimits_{a\in\tei^\star}\sum\nolimits_{b\in\tei}S_+(u)S_+(u+2), \\
			s_{11} & = & s_{12} = s_{13} = s_{14} = 2^{-\frac{3(m+1)}{2}}\sum\nolimits_{a\in\tei^\star}\sum\nolimits_{b\in\tei}S_+(u)S_+(u+1)S_+(u+2), \\
			s_{15} & = & 2^{-2(m+1)}\sum\nolimits_{a\in\tei^\star}\sum\nolimits_{b\in\tei}S_+(u)S_+(u+1)S_+(u+2)S_+(u+3).
		\end{eqnarray*}
		Then the desired result follows from Lemmas \ref{lem-S1}-\ref{lem-S1234} by solving the above system of equations on variables $x_l$ for $l=1,2,\cdots,16$. This completes the proof.
	\end{proof}

	Drawing upon the aforementioned lemmas concerning $S_+(u)$ and $S_-(u)$, we are now equipped to present the detailed proofs of our main theorems.

\section{Proof of Theorem \ref{thm1}}\label{sec-proof1}

	According to the definitions, the length of $\cc_D$ in Theorem \ref{thm1} is $|D_t|$ and the Lee weight of a codeword $\bc(u)\in\cc_D$ is $\omega_L(\bc(u))=|D_t|-(N_0-N_2)$ for any $u\in\rr$, where $D = D_t = \{x\in\tei: \Tr_1^m(x) = t\}$ and
	\[N_j = | \{ x\in D_t: \Tr_1^m(ux) = j \} | = | \{ x\in \tei: \Tr_1^m(x) = t,\ \Tr_1^m(ux) = j\} |\]
	for $j=0,2$. Moreover, one obtains
	\begin{eqnarray}\label{eq-|D|}
		| D_t | & = & \frac{1}{4}\sum_{y\in\zfour}\sum_{x\in\tei} i^{y(\Tr_1^m(x)-t)} \nonumber \\
		& = & 2^{m-2} + \frac{1}{4}\sum_{x\in\tei} i^{\Tr_1^m(x)-t} + \frac{1}{4}\sum_{x\in\tei} i^{2\Tr_1^m(x)-2t} + \frac{1}{4}\sum_{x\in\tei} i^{3\left(\Tr_1^m(x)-t\right)} \nonumber \\
		& = & 2^{m-2} + \frac{1}{2} \mathfrak{Re}\big(\sum_{x\in\tei}i^{\Tr_1^m(x)-t}\big)
	\end{eqnarray}
	and
	\begin{eqnarray}\label{eq-N}
		N_0-N_2  &=& \frac{1}{4^2} \sum_{y\in\zfour}\sum_{z\in\zfour}\sum_{x\in\tei} i^{ y(\Tr_1^m(x)-t) }i^{ z\Tr_1^m(ux) } - \frac{1}{4^2} \sum_{y\in\zfour} \sum_{z\in\zfour} \sum_{x\in\tei} i^{ y\left(\Tr_1^m(x)-t\right) }i^{ z\left(\Tr_1^m(ux)-2\right) } \nonumber  \\
		&=& \frac{1}{4^2} \sum_{y\in\zfour} \sum_{x\in\tei} i^{y\left(\Tr_1^m(x)-t\right)} \big( 1 + i^{\Tr_1^m(ux)} + i^{2\Tr_1^m(ux)} + i^{3\Tr_1^m(ux)} \big) \nonumber\\ & &-
		\frac{1}{4^2} \sum_{y\in\zfour} \sum_{x\in\tei} i^{y\left(\Tr_1^m(x)-t\right)} \big( 1 - i^{\Tr_1^m(ux)} + i^{2\Tr_1^m(ux)} - i^{3\Tr_1^m(ux)} \big)  \nonumber\\
		&=& \frac{1}{8} \sum_{x\in\tei}i^{\Tr_1^m(ux)}\sum_{y\in\zfour}i^{y\left(\Tr_1^m(x)-t\right)} + \frac{1}{8}\sum_{x\in\tei}i^{3\Tr_1^m(ux)}\sum_{y\in\zfour}i^{ y\left(\Tr_1^m(x)-t\right) }.
	\end{eqnarray}
	
	Subsequently, we will proceed to provide the proof for $\cc_D$ when $t = 0$, and it should be noted that the proofs for other values of $t$ can be derived in a similar manner.
	
	If $t=0$, then by \eqref{eq-|D|}, \eqref{eq-N}, \eqref{eq-S} and Lemma \ref{lem-chi1-0}, one gets
	\begin{eqnarray}
		\label{eq-D0} |D_0| &=& 2^{m-2} + \frac{1}{2}\mathfrak{Re}\big(\sum_{x\in\tei}i^{\Tr_1^m(x)}\big) = 2^{m-2} + 2^{\frac{m-3}{2}}\tau, \\
		\label{eq-N0-N2}  N_0-N_2 &=& \frac{1}{8}\big(S_+(u) + S_+(u+1) + S_+(u+2) + S_+(u+3)\big).
	\end{eqnarray}
	
	Thus, the length of $\cc_D$ for $t=0$ is $2^{m-2} + 2^{\frac{m-3}{2}}\tau$ and the Lee weight $\omega_L(\bc(u))=|D_0|-(N_0-N_2)$, where $u = a+2b$ with $a,b\in\tei$, can be determined as follows:
	
	\textbf{Case 1}: $a,b \in\{0,1\}$, i.e., $u\in\{0,1,2,3\}$.
	
	For this case, by Lemma \ref{lem-chi1-0} and Lemma \ref{lem-S1-dis}, one has $S_+(0)=2^{m+1}$, $S_+(2)=0$, $S_+(1)=S_+(3)=2\re(\chi_{Q_1}(0))=2^{\frac{m+1}{2}}\tau$, and then $N_0-N_2=\frac{1}{8}\big(S_+(0)+S_+(1)+S_+(2)+S_+(3)\big)=2^{m-2}+2^{\frac{m-3}{2}}\tau$, which implies that $\omega_L(\bc(u))=0 $ when $u\in\{0,1,2,3\}$.

	\textbf{Case 2}: $a \in\{0,1\}$ and $b\in\tei^\star$.
	
	If this case happens, one then obtains $N_0-N_2 = \frac{1}{8}\big(S_+(2b)+S_+(2b+1)+ S_+(2b+2)+S_+(2b+3)\big)$. By Lemma \ref{lem-S1-dis} and \eqref{eq-S}, one further has $S_+(2b) = S_+(2b+2) = 0$ due to $b\ne0,1$ and $S_+(2b+1)=S_+(2b+3)= 2\mathfrak{Re}\big(\sum_{x\in\tei} i^{\Tr_1^m(x)}(-1)^{\tr_1^m(bx)}\big)=2\re(\chi_{Q_1}(b))$ which together with Lemma \ref{lem-chi1-0} leads to
	$$N_0-N_2 =\frac{1}{2}\re(\chi_{Q_1}(b))=\frac{1}{2}\re(i^{-\Tr_1^m(b)}\cdot2^{\frac{m-1}{2}}\tau(1+i^m)).$$
	Consequently, one gets
	\[
	N_0-N_2= \left\{ \begin{array}{cl}
		2^\frac{m-3}{2}\tau, & {\rm if}\;\Tr_1^m(b)=0, \\
		2^\frac{m-3}{2}i^{m-1}\tau, & {\rm if}\;\Tr_1^m(b)=1, \\
		-2^\frac{m-3}{2}\tau, & {\rm if}\;\Tr_1^m(b)=2, \\
		2^\frac{m-3}{2}i^{m+1}\tau, & {\rm if}\;\Tr_1^m(b)=3.
	\end{array}\right.
	\]
	Note that $0\in D_0$ and $1\in D_{\rm{m\, mod}\, 4}\in\{D_1,D_3\}$ since $m$ is odd. Thus, when $b$ ranges over $\tei^\star$, the value distribution of $(N_0-N_2)$ is
	\[
	N_0-N_2= \left\{ \begin{array}{cl}
		2^\frac{m-3}{2}\tau, & |D_0|-1\; {\rm times,}  \\
		2^\frac{m-3}{2}i^{m-1}\tau, & |D_1|-\lambda \; {\rm times,} \\
		-2^\frac{m-3}{2}\tau, & |D_2| \; {\rm times,}\\
		2^\frac{m-3}{2}i^{m+1}\tau, & |D_3|-(1-\lambda) \; {\rm times,}
	\end{array}\right.
	\]
	where $\lambda=1$ if $m\equiv 1\, {\rm mod}\, 4$ and $\lambda=0$ if $m\equiv 3\, {\rm mod}\, 4$. Recall that $\tau$ is defined as in \eqref{eq-tau-sig}. Hence, the value distribution of $(N_0-N_2)$ can be written as
	\[
	N_0-N_2=\left\{ \begin{array}{cl}
		2^\frac{m-3}{2}\tau, & |D_0|+|D_1|-2\; {\rm times,}  \\
		-2^{\frac{m-3}{2}}\tau, & |D_2|+|D_3| \; {\rm times}
	\end{array}\right.
	\]
	if $m\equiv 1\, {\rm mod}\, 4$, and for $m\equiv 3\, {\rm mod}\, 4$, the distribution is
	\[
	N_0-N_2=\left\{ \begin{array}{cl}
		2^\frac{m-3}{2}\tau, & |D_0|+|D_3|-2\; {\rm times,}  \\
		-2^{\frac{m-3}{2}}\tau, & |D_1|+|D_2| \; {\rm times.}
	\end{array}\right.
	\]
	According to Lemma \ref{lem-chi1-0} and \eqref{eq-|D|}, one can obtain that $|D_0|=2^{m-2} + 2^{\frac{m-3}{2}}\tau$, $|D_1|=2^{m-2}+2^{\frac{m-3}{2}}i^{m-1}\tau$, $|D_2|=2^{m-2}-2^{\frac{m-3}{2}}\tau$ and $|D_3|=2^{m-2}- 2^{\frac{m-3}{2}}i^{m-1}\tau$. Thus, when $a$ ranges over $\{0,1\}$ and $b$ runs through $\tei^\star$, for odd $m$, the Lee weight $\omega_L(\bc(u))=|D_0|-(N_0-N_2)$ has the following distribution
	\[
	\omega_L(\bc(u))=\left\{ \begin{array}{cl}
		2^{m-2}, & 2^m+2^{\frac{m+1}{2}}\tau-4\; {\rm times,}  \\
		2^{m-2} + 2^{\frac{m-1}{2}}\tau, & 2^m-2^{\frac{m+1}{2}}\tau \; {\rm times.}
	\end{array}\right.
	\]

	\textbf{Case 3}: $a\in\tei^\star$ and $b\in\tei$.
	
	In this case, by \eqref{eq-D0}, \eqref{eq-N0-N2}, Table \ref{tab-Si} and Lemma \ref{lem-S1S2S3S4}, one can conclude that $\omega_L(\bc(u))=|D_0|-(N_0-N_2)$ has the following distribution
	\[
	\omega_L(\bc(u))=\left\{ \begin{array}{cl}
		2^{m-2}+2^\frac{m-3}{2}\tau, & x_4+x_6+x_7+x_{10}+x_{11}+x_{13} \; {\rm times,}  \\
		2^{m-2}+2^{\frac{m-3}{2}}(\tau-1), & x_2+x_3+x_5+x_9  \; {\rm times,}\\
		2^{m-2}+2^{\frac{m-3}{2}}(\tau+1), & x_8+x_{12}+x_{14}+x_{15} \; {\rm times,}  \\
		2^{m-2}+2^{\frac{m-3}{2}}(\tau-2), & x_1 \; {\rm times,}  \\
		2^{m-2}+2^{\frac{m-3}{2}}(\tau+2), & x_{16} \; {\rm times}
	\end{array}\right.
	\]
	when $a$ ranges over $\tei^\star$ and $b$ runs through $\tei$, where $x_l$'s are given as in Lemma \ref{lem-S1S2S3S4}.
	
	Based on Cases 1-3, it can be observed that $\omega_L(\bc(u))=0$ if $u\in\zfour$ and $\omega_L(\bc(u))>0$ otherwise. This implies that each codeword in $\cc_D$ appears exactly four times, specifically $\bc(u)=\bc(u+1)=\bc(u+2)=\bc(u+3)$, since $\cc_D$ is a linear code over $\zfour$. Consequently, the number of distinct codewords in $\cc_D$ is $4^{m-1}$. Therefore, the Lee weight distribution of $\cc_D$ for $t=0$ can be derived from Cases 1-3, with the frequency of each weight being divided by $4$.

	The proof for $\cc_D$ when $t=2$ follows a similar approach to the case when $t=0$. For $t\in \{1,3\}$, we can likewise determine the parameters and Lee weight distributions of $\cc_D$ by analyzing the distribution of $\big(S_+(u), iS_-(u+1), S_+(u+2), iS_-(u+3)\big)$, using Lemma \ref{lem-S1}, Lemma \ref{lem-S12} (1)(3), Lemma \ref{lem-S123} (2), Lemma \ref{lem-S1234} (2), as we did in Lemma \ref{lem-S1S2S3S4} for $\big(S_+(u), S_+(u+1), S_+(u+2), S_+(u+3)\big)$. To avoid redundancy, we have omitted the detailed proofs for other values of $t$.
	
	This completes the proof.

\section{Proof of Theorem \ref{thm2}}\label{sec-proof2}
	
	It can be readily verified from the definitions that the length of \(\cc_D\) in Theorem \ref{thm2} is $|D| = |D_{t_1}|+|D_{t_2}|$ and the Lee weight of a codeword $\bc(u)\in\cc_D$ is $\omega_L(\bc(u))=|D|-(N_0-N_2)$ for any $u\in\rr$, where $D = D_{t_1} \cup D_{t_2} = \{x\in\tei: \Tr_1^m(x) = t_1 {\rm \ or}\ \Tr_1^m(x) = t_2\}$ and
	$$N_j = | \{ x\in D: \Tr_1^m(ux) = j \} | = | \{ x\in \tei: \Tr_1^m(x) = t_1 {\rm \ or}\ \Tr_1^m(x) = t_2,\ \Tr_1^m(ux) = j\} |$$
	for $j=0,2$. We will now proceed to demonstrate the proof for $\cc_D$ when $(t_1, t_2) = (0, 2)$, and it is worth mentioning that the proof for $(t_1, t_2) = (1, 3)$ can be conducted analogously.
	If $(t_1, t_2) = (0, 2)$, then by \eqref{eq-|D|}, \eqref{eq-N}, \eqref{eq-S}, and Lemma \ref{lem-chi1-0}, we obtain
	\begin{eqnarray}
		|D| &=& | D_0 | + | D_2 | = 2^{m-1}, \nonumber \\
		\label{eq2-N0-N2} N_0-N_2 & = & \frac{1}{8} \big(S_+(u) + S_+(u+1) + S_+(u+2) + S_+(u+3)\big) \nonumber \\
		&&+ \frac{1}{8} \big(S_+(u) - S_+(u+1) + S_+(u+2) - S_+(u+3)\big) \nonumber \\
		& = & \frac{1}{4}\big(S_+(u)+S_+(u+2)\big).
	\end{eqnarray}
	
	Let $u=a+2b$ with $a,b\in\tei$. Then the Lee weight $\omega_L(\bc(u))=2^{m-1}-(N_0-N_2)$ can be determined as follows:
	
	\textbf{Case 1}: $a = 0$ and $b\in\{0, 1\}$, i.e. $u\in\{0, 2\}$.
	
	From Lemma \ref{lem-S1-dis}, we have $S_+(0) = 2^{m+1}$ and $S_+(2) = 0 $. Consequently, $N_0 - N_2 = \frac{1}{4}\big(S_+(0) + S_+(2)\big) = 2^{m-1}$, which indicates that $\omega_L(\bc(u)) = 0$ when $u \in \{0, 2\}$.
	
	\textbf{Case 2}: $a = 0$ and $b\in\tei^\star$.
	
	In this case, we have $N_0 - N_2 = \frac{1}{4}\big(S_+(2b) + S_+(2b+2)\big)$. According to Lemma \ref{lem-S1-dis}, for any $b \in \tei^\star$, $S_+(2b) = S_+(2b+2) = 0$. As a result, $N_0 - N_2 = 0 $ and $\omega_L(\bc(u)) = 2^{m-1} $ for any $u = 2b$ with $b \in \tei^\star $.

	\textbf{Case 3}: $a = 1$ and $b\in\tei$.
	
	By \eqref{eq2-N0-N2}, for this case, we have $N_0-N_2 = \frac{1}{4}\big(S_+(2b+1)+S_+(2b+3)\big)$. According to \eqref{eq-S}, we obtain $S_+(2b+1)=S_+(2b+3)= 2\mathfrak{Re}\big(\sum_{x\in\tei} i^{\Tr_1^m(x)}(-1)^{\tr_1^m(bx)}\big)=2\re(\chi_{Q_1}(b))$, which, in conjunction with Lemma \ref{lem-chi1-0}, yields
	$N_0-N_2 = \re(\chi_{Q_1}(b)) = \re(i^{-\Tr_1^m(b)}\cdot2^{\frac{m-1}{2}}\tau(1+i^m))$.
	This implies that
	\[
	N_0-N_2= \left\{ \begin{array}{cl}
		2^\frac{m-1}{2}\tau, & |D_0|\; {\rm times,}  \\
		2^\frac{m-1}{2}i^{m-1}\tau, & |D_1|\; {\rm times,} \\
		-2^\frac{m-1}{2}\tau, & |D_2|\; {\rm times,}\\
		2^\frac{m-1}{2}i^{m+1}\tau, & |D_3|\; {\rm times}
	\end{array}\right.
	\]
	when $b$ ranges over $\tei$. By the definition of $\tau$ in \eqref{eq-tau-sig}, one has when $b$ ranges over $\tei$, the value distribution of $(N_0-N_2)$ is
	\[
	N_0-N_2=\left\{ \begin{array}{cl}
		2^\frac{m-1}{2}\tau, & |D_0|+|D_1|\; {\rm times,}  \\
		-2^{\frac{m-1}{2}}\tau, & |D_2|+|D_3| \; {\rm times}
	\end{array}\right.
	\]
	if $m\equiv 1\, {\rm mod}\, 4$, and for $m\equiv 3\, {\rm mod}\, 4$, the distribution is
	\[
	N_0-N_2=\left\{ \begin{array}{cl}
		2^\frac{m-1}{2}\tau, & |D_0|+|D_3|\; {\rm times,}  \\
		-2^{\frac{m-1}{2}}\tau, & |D_1|+|D_2| \; {\rm times.}
	\end{array}\right.
	\]
	Additionally, the values of $|D_t|$ for $1\le t \le 4$ can be derived from \eqref{eq-|D|}. Consequently, we find that $N_0-N_2 = \pm2^\frac{m-1}{2}\tau$, which occurs $2^{m-1} \pm 2^\frac{m-1}{2}\tau$ times respectively. When $a=1$ and $b$ runs through $\tei$ for odd $m$, the Lee weight $\omega_L(\bc(u))=2^{m-1}-(N_0-N_2)$ thus has the following distribution
	$$ \omega_L(\bc(u)) = \left\{ \begin{array}{cl}
		2^{m-1}-2^\frac{m-1}{2}\tau, & 2^{m-1} + 2^\frac{m-1}{2}\tau\ {\rm times,} \\
		2^{m-1}+2^\frac{m-1}{2}\tau, & 2^{m-1} - 2^\frac{m-1}{2}\tau\ {\rm times.}
	\end{array}\right.$$	
	
	\textbf{Case 4}: $a\in\tei^\star$ and $b\in\tei$.
	
	In this case, leveraging \eqref{eq2-N0-N2} and Lemma \ref{lem-S1S2}, it can be deduced that $\omega_L(\bc(u)) = 2^{m-1}-(N_0-N_2)$ exhibits the following distribution
	\[
	\omega_L(\bc(u))=\left\{ \begin{array}{cl}
		2^{m-1}, & 2^{2m-1}-2^m \; {\rm times,}  \\
		2^{m-1}-2^{\frac{m-1}{2}}, & (2^m-2)(2^{m-2}+2^\frac{m-1}{2}) \; {\rm times,}\\
		2^{m-1}+2^{\frac{m-1}{2}}, & (2^m-2)(2^{m-2}-2^\frac{m-1}{2}) \; {\rm times}
	\end{array}\right.
	\]
	when $a$ ranges over $\tei^\star$ and $b$ runs through $\tei$.
	
	As observed from the previous discussion, $\omega_L(\bc(u)) > 0$ holds for any $u \in \tei \setminus \{0, 2\}$. This suggests that each codeword in $\cc_D$ appears exactly twice, specifically $\bc(u) = \bc(u+2)$, since $\cc_D$ is a linear code over $\mathbb{Z}_4$. Hence, the number of distinct codewords in $\cc_D$ is $2^{2m-1}$. Accordingly, the Lee weight distribution of $\cc_D$ for $(t_1, t_2) = (0, 2)$ can be deduced from Cases 1-4, with the frequency of each weight being halved.
	
	This completes the proof.

\section{Proof of Theorem \ref{thm3}}\label{sec-proof3}
	
	It is straightforward to confirm, using the definitions, that the length of \(\cc_D\) in Theorem \ref{thm3} is $|D| = 2^m - |D_t|$, and the Lee weight of a codeword $\bc(u)\in\cc_D$ is $\omega_L(\bc(u)) = 2^m - |D_t| -(N_0-N_2)$ for any $u\in\rr$, where $D = \tei\backslash D_t = \{x\in\tei: \Tr_1^m(x) \ne t\}$ and
	$$N_j = | \{ x\in D: \Tr_1^m(ux) = j \} | = | \{ x\in \tei: \Tr_1^m(x) \ne t,\ \Tr_1^m(ux) = j\} |$$
	for $j=0,2$. We will now present the proof for $\cc_D$ when $t = 0$,  and the proofs for other values of $t$ can be obtained in a similar manner.
	
	If $t = 0$, utilizing \eqref{eq-S}, \eqref{eq-|D|}, and \eqref{eq-N0-N2}, we arrive at
	\begin{eqnarray}
		\label{eq3-D0} |D| &=& 2^m-|D_0| = 3\cdot2^{m-2} - 2^\frac{m-3}{2}\tau, \\
		\label{eq3-N0-N2} N_0-N_2 & = & \frac{1}{8}\sum_{j=1}^{3}\big( \sum_{x\in\tei}i^{\Tr_1^m(ux)}\sum_{y\in\zfour}i^{y\left(\Tr_1^m(x)-j\right)} + \sum_{x\in\tei}i^{3\Tr_1^m(ux)}\sum_{y\in\zfour}i^{ y\left(\Tr_1^m(x)-j\right) } \big) \nonumber \\
		& = & \frac{1}{8}\big( 3S_+(u) - S_+(u+1) - S_+(u+2) - S_+(u+3) \big).
	\end{eqnarray}
	
	Therefore, the length of $\cc_D$ for $t = 0$ is $|D| = 3\cdot2^{m-2} - 2^\frac{m-3}{2}\tau$, and the Lee weight $\omega_L(\bc(u)) = |D|-(N_0-N_2)$, where $u = a+2b$ with $a,b\in\tei$, can be computed as follows:
	
	\textbf{Case 1}: $a = 0$ and $b=0$, i.e., $u=0$.
	
	This is a trivial case where we have $\omega_L(\bc(0))=0$.
	
	\textbf{Case 2}: $a = 0$ and $b=1$, i.e., $u=2$.
	
	Based on Lemmas \ref{lem-chi1-0}-\ref{lem-S1-dis}, we obtain $S_+(0) = 2^{m+1}$, $S_+(2) = 0$, and $ S_+(1) = S_+(3) = 2\re(\chi_{Q_1}(0)) = 2^\frac{m+1}{2}\tau$.
	Consequently, by \eqref{eq3-N0-N2}, we have $N_0-N_2 = \frac{1}{8}\big( 3S_+(2) - S_+(3) - S_+(0) - S_+(1) \big) = -2^{m-2} - 2^\frac{m-3}{2}\tau$, which implies that
	$\omega_L(\bc(2)) = 2^m $.
	
	\textbf{Case 3}: $a = 0$ and $b\in\tei^\star$.
	
	In this case, by \eqref{eq3-N0-N2}, we have $N_0-N_2 = \frac{1}{8}\big(3S_+(2b)-S_+(2b+1)-S_+(2b+2)-S_+(2b+3)\big)$.  According to Lemma \ref{lem-S1-dis} and \eqref{eq-S}, we further obtain $ S_+(2b) = S_+(2b+2) = 0$ since $b\ne0,1$, and $S_+(2b+1) = S_+(2b+3) = 2\re(\chi_{Q_1}(b))$, which leads to
	$$N_0-N_2 = -\frac{1}{2}\re(\chi_{Q_1}(b)) = -\frac{1}{2}\re(i^{-\Tr_1^m(b)}\cdot2^{\frac{m-1}{2}}\tau(1+i^m)).$$
	Analogous to Case 2 in the proof of Theorem \ref{thm1}, it is found that as	$b$ ranges over $\tei^\star$, the value distribution of $(N_0-N_2)$ is
	\[
	N_0-N_2=
	\left\{ \begin{array}{cl}
		-2^\frac{m-3}{2}\tau, & 2^{m-1}+2^\frac{m-1}{2}\tau-2\; {\rm times,}  \\
		2^{\frac{m-3}{2}}\tau, & 2^{m-1}-2^\frac{m-1}{2}\tau \; {\rm times.}
	\end{array}\right.
	\]
	Therefore, when $a=0$ and $b$ runs through $\tei^\star$, for odd $m$, the Lee weight $\omega_L(\bc(u))=|D|-(N_0-N_2)$ exhibits the following distribution
	\[
	\omega_L(\bc(u))=\left\{ \begin{array}{cl}
		3\cdot2^{m-2}, & 2^{m-1}+2^\frac{m-1}{2}\tau-2\; {\rm times,}  \\
		3\cdot2^{m-2} - 2^\frac{m-1}{2}\tau, & 2^{m-1}-2^\frac{m-1}{2}\tau \; {\rm times.}
	\end{array}\right.
	\]
	
	\textbf{Case 4}: $a = 1$ and $b\in\{0,1\}$, i.e., $u\in\{1, 3\}$.
	
	By Lemmas \ref{lem-chi1-0}-\ref{lem-S1-dis}, we obtain $S_+(0) = 2^{m+1}$, $S_+(2) = 0$, and $ S_+(1) = S_+(3) = 2\re(\chi_{Q_1}(0)) = 2^\frac{m+1}{2}\tau$.
	Therefore, by \eqref{eq3-N0-N2}, $N_0-N_2 = \frac{1}{8}\big(3S_+(1) - S_+(2) - S_+(3) - S_+(0)\big) = 2^\frac{m-3}{2}\tau - 2^{m-2}$, which implies that
	$\omega_L(\bc(u)) = |D| - (N_0-N_2) = 2^m - 2^\frac{m-1}{2}\tau$ when $u\in\{1,3\}$.
	
	\textbf{Case 5}: $a = 1$ and $b\in\tei^\star$.
	
	In this case, by \eqref{eq3-N0-N2}, we have $N_0-N_2 = \frac{1}{8}\big(3S_+(2b+1)-S_+(2b+2)-S_+(2b+3)-S_+(2b)\big)$. According to  Lemma \ref{lem-S1-dis} and \eqref{eq-S}, we obtain  $S_+(2b+1) = S_+(2b+3) = 2\re(\chi_{Q_1}(b))$ and $ S_+(2b) = S_+(2b+2) = 0$ since $b\ne0,1$, which results in $N_0-N_2 = \frac{1}{2}\re(\chi_{Q_1}(b)).$  Similar to Case 3, we find that when $a=1$, $b$ runs through $\tei^\star$, and $m$ is odd, the distribution of $N_0-N_2$ is
	\[
	N_0-N_2=
	\left\{ \begin{array}{cl}
		2^\frac{m-3}{2}\tau, & 2^{m-1}+2^\frac{m-1}{2}\tau-2\; {\rm times,}  \\
		-2^{\frac{m-3}{2}}\tau, & 2^{m-1}-2^\frac{m-1}{2}\tau \; {\rm times,}
	\end{array}\right.
	\]
	and consequently, the distribution of Lee weights $\omega_L(\bc(u))=|D|-(N_0-N_2)$ is
	\[
	\omega_L(\bc(u))=\left\{ \begin{array}{cl}
		3\cdot2^{m-2}, & 2^{m-1}-2^\frac{m-1}{2}\tau \; {\rm times,}  \\
		3\cdot2^{m-2} - 2^\frac{m-1}{2}\tau, & 2^{m-1}+2^\frac{m-1}{2}\tau-2 \; {\rm times.}
	\end{array}\right.
	\]
	
	\textbf{Case 6}: $a\in\tei^\star$ and $b\in\tei$.
	
	In this case, by \eqref{eq3-D0}, \eqref{eq3-N0-N2}, and Lemma \ref{lem-S1S2S3S4}, we can deduce that the distribution of $\omega_L(\bc(u))=|D|-(N_0-N_2)$ is
	\[
	\omega_L(\bc(u))=\left\{ \begin{array}{cl}
		3\cdot2^{m-2} - 2^\frac{m-3}{2}\tau, & x_1+x_{16} \; {\rm times,}  \\
		3\cdot2^{m-2}-2^\frac{m-3}{2}(\tau+1), & x_2+x_3+x_5  \; {\rm times,}\\
		3\cdot2^{m-2}-2^\frac{m-3}{2}(\tau-1), & x_{12}+x_{14}+x_{15} \; {\rm times,}  \\
		3\cdot2^{m-2}-2^\frac{m-3}{2}(\tau+2), & x_4+x_6+x_7 \; {\rm times,}  \\
		3\cdot2^{m-2}-2^\frac{m-3}{2}(\tau-2), & x_{10}+x_{11}+x_{13} \; {\rm times,}  \\
		3\cdot2^{m-2}-2^\frac{m-3}{2}(\tau+3), & x_8 \; {\rm times,}  \\
		3\cdot2^{m-2}-2^\frac{m-3}{2}(\tau-3), & x_9 \; {\rm times}
	\end{array}\right.
	\]
	when $a$ ranges over $\tei^\star$ and $b$ runs through $\tei$, where $x_l$'s are given as in Lemma \ref{lem-S1S2S3S4}.
	
	Based on Cases 1-6, it is clear that $\omega_L(\bc(u))>0$ for any nonzero element $u\in\rr$. Therefore, the number of distinct codewords in $\cc_D$ is $4^m$, and the Lee weight distribution of $\cc_D$ for $t=0$ can be derived from the aforementioned cases.
	
	The proof for $\cc_D$ when $t=2$ follows a similar approach to the proof for the case when $t=0$. For the case where $t\in \{1,3\}$, we can similarly determine the parameters and Lee weight distributions of $\cc_D$ by analyzing the distribution of $\big(S_+(u), iS_-(u+1), S_+(u+2), iS_-(u+3)\big)$ with the aid of Lemma \ref{lem-S1}, Lemma \ref{lem-S12} (1)(3), Lemma \ref{lem-S123} (2), and Lemma \ref{lem-S1234} (2). To avoid redundancy, we will not provide the detailed proofs for other values of $t$.
	
	This completes the proof.

\section{Conclusion}\label{sec-conclusion}
	
	In this paper, we have explored several classes of few-weight $\mathbb{Z}_4$-linear codes in the form \eqref{eq-code}, by selecting $D = D_{t_1}$, $D = D_{t_1} \cup D_{t_2}$, and $D = D_{t_1} \cup D_{t_2} \cup D_{t_3}$ with $t_1, t_2, t_3 \in \mathbb{Z}_4$ as the defining sets, where $D_t$ is defined as in \eqref{eq-Dt} for $t \in \mathbb{Z}_4$. We initially employed the properties of the trace function over $\mathbb{Z}_4$ to examine the value distributions of two classes of exponential sums and their related properties. Building on these findings, we subsequently analyzed the parameters and Lee weight distributions of the $\mathbb{Z}_4$-linear codes we constructed for odd $m$. Furthermore, by comparing with the existing database of $\mathbb{Z}_4$ codes, examples with small $m$ provided in Table \ref{tab-newcode} demonstrate that many new $\mathbb{Z}_4$-linear codes and $\mathbb{Z}_4$-linear codes with the best known minimum Lee distance were identified from our codes. As future work, we cordially invite the reader to determine the Lee weight distribution of $\cc_D$ for even $m$.

\section*{Acknowledgements}
This work  was supported by the National Natural Science Foundation of China under Grant 12471492,  and by the Innovation Group Project of the Natural Science Foundation of Hubei Province of China under Grant 2023AFA021.

\appendices
\section {Proof of the Lemma \ref{lem-S123}}\label{append-a}	
	
	This section is devoted to proving Lemma \ref{lem-S123}. We present the proof for Case (1) here, as the other case can be proven using the same approach. For simplicity, define
	\[\Omega= \sum_{b\in \tei}i^{\Tr_1^m\left(2b(x\oplus y\oplus z)\right)}\]
	and
	$$\begin{aligned}
		& \rho_1 = \sum_{x,y,z\in\tei}i^{\Tr_1^m(2y+z)}\sum_{a\in \tei^\star}i^{\Tr_1^m\left(a(x+y+z)\right)}\Omega, \! &
		& \rho_2 = \sum_{x,y,z\in\tei}i^{\Tr_1^m(2y+3z)}\sum_{a\in \tei^\star}i^{\Tr_1^m\left(a(x+y+3z)\right)}\Omega, && \\
		& \rho_3 = \sum_{x,y,z\in\tei}i^{\Tr_1^m(2y+z)}\sum_{a\in \tei^\star}i^{\Tr_1^m\left(a(x+3y+z)\right)}\Omega,\! &
		& \rho_4 = \sum_{x,y,z\in\tei}i^{\Tr_1^m(2y+3z)}\sum_{a\in \tei^\star}i^{\Tr_1^m\left(a(x+3y+3z)\right)}\Omega. && \\
	\end{aligned}$$
	Then, by \eqref{eq-S}, substituting $ u = a+2b $, a straightforward computation gives
	\begin{eqnarray}\label{eq-S123}
		\sum_{a\in\tei^\star}\sum_{b\in\tei}S_+(u)S_+(u+2)S_+(u+1) =\rho_1+\rho_2+\rho_3+\rho_4+\overline{\rho_1}+\overline{\rho_2}+\overline{\rho_3}+\overline{\rho_4},
	\end{eqnarray}
	where $\overline{\rho_j}$ denotes the conjugate of $\rho_j$.
	
	Note that $\Omega= \sum_{b\in \tei}(-1)^{\tr_1^m\left(b(x\oplus y\oplus z)\right)} = 2^m$ if $z = x\oplus y$ and $0$ otherwise. This gives
	$$ \rho_1 = 2^m\sum_{x,y\in\tei}i^{\Tr_1^m(x+3y+2\sqrt{xy})}\sum_{a\in \tei^\star}(-1)^{\tr_1^m(a(x\oplus y\oplus \sqrt{xy}))}. $$
	Then, by the identity $(\sqrt{x}\oplus \sqrt{y})(x\oplus y\oplus \sqrt{xy})=\sqrt{x}^3\oplus \sqrt{y}^3$, we can conclude that $x\oplus y\oplus \sqrt{xy}=0$ if and only if $\sqrt{x}=\xi\sqrt{y}$, where $\xi$ belongs to some extension field of $\tei$, $\xi\ne 1$, and $\xi^3=1$. This implies that  $x\oplus y\oplus \sqrt{xy}=0$ if and only if $x=y=0$ since $m$ is odd. Hence, we have
	$$\sum_{a\in \tei^\star}i^{\tr_1^m(a(x\oplus y\oplus \sqrt{xy}))} = \left\{\begin{array}{cl}
		2^m-2,\ & {\rm if\ } x = y = 0, \\
		-1-(-1)^{\tr_1^m(x\oplus y\oplus \sqrt{xy})},\ & {\rm otherwise.}
	\end{array}\right.$$
	Moreover, using $2\tr_1^m(x)=2\Tr_1^m(x)$, $\Tr_1^m(x^{2^k}) = \Tr_1^m(x)$, and $\Tr_1^m(x+y)=\Tr_1^m(x)+\Tr_1^m(y)$ for a positive $k$ and $x,y\in\tei$, we obtain
	\begin{eqnarray*}
		\rho_1 & = & 2^m(2^m-2) - 2^m\sum_{x,y\in\tei,(x,y)\ne(0,0)}i^{\Tr_1^m(x+3y+2\sqrt{xy})}\big(1+i^{2\Tr_1^m(x\oplus y\oplus \sqrt{xy})}\big)
		\\ & = & 2^{2m} - 2^m\sum_{x,y\in\tei}i^{\Tr_1^m(x+3y+2xy)}\big(1+i^{2\Tr_1^m(x+y+xy)}\big)
		\\ & = & 2^{2m} - 2^m\sum_{x,y\in\tei}i^{\Tr_1^m(x+3y+2xy)} - 2^m\sum_{x\in\tei}i^{\Tr_1^m(3x)}\sum_{y\in\tei}i^{\Tr_1^m(y)}
		\\ & =& 2^{2m} - 2^m\sum_{x\in\tei}i^{\Tr_1^m(x)} \sum_{y\in\tei}i^{\Tr_1^m\left(y+2(1\oplus x)y\right)}- 2^m\sum_{x\in\tei}i^{\Tr_1^m(3x)}\sum_{y\in\tei}i^{\Tr_1^m(y)}
		\\ & =& 2^{2m} - 2^m\sum_{x\in\tei}i^{\Tr_1^m(x)}i^{-\Tr_1^m(1\oplus x)}\chi_{Q_1}(0)-2^m i^{-\Tr_1^m(1)}\chi_{Q_1}(0)^2
		\\ & =& 2^{2m} - 2^mi^{\Tr_1^m(3)}\chi_{Q_1}(0)\sum_{x\in\tei}(-1)^{\tr_1^m(\sqrt{x})}-2^{2m}
		\\ & =& 0,
	\end{eqnarray*}
	where the third-to-last equality follows from Lemma \ref{lem-chi1-0}, and the second-to-last equality follows from the fact that $m$ is odd.
	
	In a similar manner, we can obtain that
	$\rho_2 = 2^{2m+1}(2^\frac{m-1}{2}\tau-1)$, and $\rho_3=\rho_4=2^{2m}(2^\frac{m-1}{2}\tau-1)(1-i^m)$. Substituting the values of $\rho_j$ for $1\le j \le 4$ into \eqref{eq-S123} yields
	$$\sum_{a\in\tei^\star}\sum_{b\in\tei} S_+(u)S_+(u+2)S_+(u+1) = 2\rho_2+4\re(\rho_3) = 2^{2m+3}(2^\frac{m-1}{2}\tau-1),$$
	which completes the proof.

\section {Proof of the Lemma \ref{lem-S1234}}\label{append-b}	
	
	This section is dedicated to proving Lemma \ref{lem-S1234}.	For brevity, we present the proof for Case (1) only, as the remaining case can be demonstrated using a similar approach. For the sake of simplicity, define
	\[\Delta= \sum_{b\in\tei}i^{\Tr_1^m\big(2b(x\oplus y\oplus z\oplus v)\big)}\]
	which equals  $2^m$ if $y = x\oplus z \oplus v$, and $0$ otherwise, and let
	$$\begin{aligned}	
		& \theta_1 = \sum_{\substack{a\in \tei^\star\\x,y,z,v\in\tei}} i^{\Tr_1^m\left( 2y+z+3v+a(x+y+z+v) \right)} \Delta, \! &
		& \theta_2 = \sum_{\substack{a\in \tei^\star\\x,y,z,v\in\tei}} i^{\Tr_1^m\left( 2y+z+v+a(x+y+z+3v) \right)} \Delta, && \\
		& \theta_3 = \sum_{\substack{a\in \tei^\star\\x,y,z,v\in\tei}} i^{\Tr_1^m\left( 2y+3z+3v+a(x+y+3z+v) \right)} \Delta,\! &
		& \theta_4 = \sum_{\substack{a\in \tei^\star\\x,y,z,v\in\tei}} i^{\Tr_1^m\left( 2y+3z+v+a(x+y+3z+3v) \right)} \Delta, && \\
		& \theta_5 = \sum_{\substack{a\in \tei^\star\\x,y,z,v\in\tei}} i^{\Tr_1^m\left( 2y+z+3v+a(x+3y+z+v) \right)} \Delta,\! &
		& \theta_6 = \sum_{\substack{a\in \tei^\star\\x,y,z,v\in\tei}} i^{\Tr_1^m\left( 2y+z+v+a(x+3y+z+3v) \right)} \Delta, && \\
		& \theta_7 = \sum_{\substack{a\in \tei^\star\\x,y,z,v\in\tei}} i^{\Tr_1^m\left( 2y+3z+3v+a(x+3y+3z+v) \right)} \Delta,\! &
		& \theta_8 = \sum_{\substack{a\in \tei^\star\\x,y,z,v\in\tei}} i^{\Tr_1^m\left( 2y+3z+v+a(x+3y+3z+3v) \right)} \Delta.
	\end{aligned}$$	
	Substituting $u = a+2b$ and applying \eqref{eq-S}, a straightforward calculation yields
	$$ \sum_{a\in\tei^\star}\sum_{b\in\tei} S_+(u)S_+(u+2)S_+(u+1)S_+(u+3) = 2\re ( \theta_1 + \theta_2 + \theta_3 + \theta_4 + \theta_5 + \theta_6 + \theta_7 + \theta_8 ). $$
	
	Next, we aim to compute the value of $\theta_1$. By the definition of $\Delta$, we have
	\begin{eqnarray}
		\theta_1 &=& 2^m \sum_{x,z,v\in\tei}i^{\Tr_1^m(2x+3z+v)}\sum_{a\in \tei^\star}i^{\Tr_1^m(2a\varepsilon)} \nonumber
		\\ \label{eq-theta1} &=& 2^m(2^m-2)\sum_{\substack{x,z,v\in\tei\\ \varepsilon=0}}i^{\Tr_1^m(2x+3z+v)} - 2^m\sum_{\substack{x,z,v\in\tei\\ \varepsilon\ne0}}i^{\Tr_1^m(2x+3z+v)}(1+i^{\Tr_1^m(2\varepsilon)})
	\end{eqnarray}
	since the inner sum equals $2^m-2$ if $\varepsilon = 0$, and $-1-i^{\Tr_1^m(2\varepsilon)}$ if $\varepsilon \ne 0$, where $\varepsilon=x\oplus z\oplus v\oplus \sqrt{xz}\oplus \sqrt{xv}\oplus \sqrt{zv}$. Observe that $\varepsilon $ runs through $\tei$ when $x,z,v$ range over $\tei$. Moreover, $\varepsilon=0$ if and only if $\varepsilon^2=(x\oplus\xi z\oplus\xi^2v)(x\oplus\xi^2z\oplus\xi v)=0$,  where $\xi$ is an element in some extension field of $\tei$ satisfying $\xi\ne 1$ and $\xi^3=1$. We can then conclude that $\varepsilon=0$ if and only if $x=z=v$. Specifically, if $x\oplus\xi z\oplus\xi^2v=0$, we have $(x\oplus v)\oplus(z\oplus v)\xi=0$ due to the fact that $\xi^2\oplus\xi\oplus 1=0$. This implies $z=v$ and subsequently $x=v$ since $m$ is odd and $\xi\in \mathbb{F}_{2^2}\not\subset \tei$. Similarly, if $x\oplus\xi^2 z\oplus\xi v=0$, we also obtain $x=z=v$. Therefore, by \eqref{eq-theta1}, and using the properties $2\tr_1^m(x)=2\Tr_1^m(x)$, $\Tr_1^m(x^{2^k}) = \Tr_1^m(x)$, and $\Tr_1^m(x+y)=\Tr_1^m(x)+\Tr_1^m(y)$ for a positive $k$ and $x,y\in\tei$, we have
	\begin{eqnarray*}
		\theta_1 &=& 2^m(2^m-1)\sum_{\substack{x,z,v\in\tei\\ \varepsilon=0}}i^{\Tr_1^m(2x+3z+v)} - 2^m\sum_{\substack{x,z,v\in\tei\\ \varepsilon\ne0}}i^{\Tr_1^m(z+3v+2\sqrt{xz}+2\sqrt{xv}+2\sqrt{zv})}\\
		&=&  2^m(2^m-1)\sum_{\substack{x\in\tei}}i^{\Tr_1^m(2x)} - 2^m\sum_{\substack{x,z,v\in\tei\\ \varepsilon\ne0}}i^{\Tr_1^m(z+3v+2xz+2xv+2zv)}    \\
		&=& - 2^m\big(\sum_{x,z,v\in\tei}i^{\Tr_1^m(z+3v+2xz+2xv+2zv)}-\sum_{\substack{x,z,v\in\tei\\ \varepsilon=0}}i^{\Tr_1^m(z+3v+2xz+2xv+2zv)}\big)   \\
		&=& - 2^m\big(\sum_{z,v\in\tei}i^{\Tr_1^m(z+3v+2zv)}\sum_{x\in\tei}i^{\Tr_1^m(2(z\oplus v)x)}-\sum_{x\in\tei}i^{\Tr_1^m(2x^2)}\big)   \\
		&=& - 2^{2m}\sum_{z\in\tei}i^{\Tr_1^m(2z^2)}  \\
		&=& 0.
	\end{eqnarray*}
	
	By analogous reasoning, we can derive that $\theta_4=\theta_6 =\theta_7=0$. Next, we proceed to determine the value of $\theta_2$.
	Notice that $\Delta = 2^m$ if $v=x\oplus y\oplus z$, and $0$ otherwise. Similar to the calculation of $\theta_1$, substituting $v=x\oplus y\oplus z$ into $\theta_2$  yields
	\begin{eqnarray}\label{eq-theta2}
		\theta_2 &=& 2^m \sum_{x,y,z\in\tei}i^{\Tr_1^m(x+3y+2z+2\epsilon)}\sum_{a\in \tei^\star}i^{\Tr_1^m(2a\epsilon)} \nonumber
		\\  &=& 2^m(2^m-2)\sum_{\substack{x,y,z\in\tei\\ \epsilon=0}}i^{\Tr_1^m(x+3y+2z)} - 2^m\sum_{\substack{x,y,z\in\tei\\ \epsilon\ne0}}i^{\Tr_1^m(x+3y+2z+2\epsilon)}(1+i^{\Tr_1^m(2\epsilon)}) \nonumber
		\\  &=& 2^m(2^m-2)\sum_{\substack{x,y,z\in\tei\\ \epsilon=0}}i^{\Tr_1^m(x+3y+2z)} - 2^m\sum_{\substack{x,y,z\in\tei\\ \epsilon\ne0}}i^{\Tr_1^m(x+3y+2z)}(1+i^{\Tr_1^m(2\epsilon)}),
	\end{eqnarray}
	where the second equality holds because the inner sum equals $2^m-2$ if $\epsilon=0$, and $-1-i^{\Tr_1^m(2\epsilon)}$ if $\epsilon\ne 0$, with $\epsilon=\sqrt{xy}\oplus \sqrt{xz}\oplus \sqrt{yz}$. By leveraging the following two identities
	\begin{eqnarray*}
		\sum_{x,y,z\in\tei}i^{\Tr_1^m(x+3y+2z)}
		&=& \sum_{\substack{x,y,z\in\tei\\ \epsilon=0}}i^{\Tr_1^m(x+3y+2z)}+\sum_{\substack{x,y,z\in\tei\\ \epsilon\ne0}}i^{\Tr_1^m(x+3y+2z)} \\
		\sum_{x,y,z\in\tei}i^{\Tr_1^m(x+3y+2z+2\epsilon)}
		&=& \sum_{\substack{x,y,z\in\tei\\ \epsilon=0}}i^{\Tr_1^m(x+3y+2z+2\epsilon)}+\sum_{\substack{x,y,z\in\tei\\ \epsilon\ne0}}i^{\Tr_1^m(x+3y+2z+2\epsilon)}
	\end{eqnarray*}
	and by \eqref{eq-theta2}, $\theta_2$ can be expressed as
	\begin{eqnarray}
		\theta_2 &=& 2^{2m}\sum_{\substack{x,y,z\in\tei\\ \epsilon=0}}i^{\Tr_1^m(x+3y+2z)} - 2^m\sum_{x,y,z\in\tei}i^{\Tr_1^m(x+3y+2z)} - 2^m\sum_{x,y,z\in\tei}i^{\Tr_1^m(x+3y+2z+2\epsilon)} \nonumber
		\\ &=& 2^{2m}\sum_{\substack{x,y,z\in\tei\\ \epsilon=0}}i^{\Tr_1^m(x+3y+2z)} - 2^m\sum_{x,y\in\tei} i^{\Tr_1^m(x+3y)}\sum_{z\in\tei} i^{\Tr_1^m(2z)} - 2^m\sum_{x,y,z\in\tei} i^{\Tr_1^m(x+3y+2z+2\epsilon)} \nonumber
		\\ \label{eq-theta2-1}
		&=& 2^{2m}\sum_{\substack{x,y,z\in\tei\\ \epsilon=0}}i^{\Tr_1^m(x+3y+2z)} -  2^m\sum_{x,y,z\in\tei}i^{\Tr_1^m(x+3y+2z+2\epsilon)}.
	\end{eqnarray}
	Then, by utilizing the properties of the trace function over $\rr$, we have $\Tr_1^m(x+3y+2z+2\epsilon) = \Tr_1^m(x+3y+2z+2xy+2xz+2yz) = \Tr_1^m(x+3y+2xy)+\Tr_1^m(2z(1\oplus x\oplus y))$, which leads to
	\begin{eqnarray*}
		\sum_{x,y,z\in\tei}i^{\Tr_1^m(x+3y+2z+2\epsilon)}
		&=& \sum_{x,y\in\tei}i^{\Tr_1^m(x+3y+2xy)}\sum_{z\in\tei}i^{\Tr_1^m(2z(1\oplus x\oplus y))} \\
		&=& 2^m\sum_{y\in\tei}i^{\Tr_1^m\left((1\oplus y)+3y+2y(1\oplus y)\right)} \\
		&=& 2^m\sum_{y\in\tei}i^{\Tr_1^m(1+2\sqrt{y}+2y+2y^2)}\\
		&=& 2^m\sum_{y\in\tei}i^{\Tr_1^m(1+2y)}\\
		&=& 0,
	\end{eqnarray*}
	where the second-to-last equality follows from $2\Tr_1^m(\sqrt{y})=2\Tr_1^m(y)=2\Tr_1^m(y^2)$. This together with \eqref{eq-theta2-1} gives
	\begin{equation}\label{eq-theta2-2}
		\theta_2 = 2^{2m}\sum_{\substack{x,y,z\in\tei\\ \epsilon=0}}i^{\Tr_1^m(x+3y+2z)}.
	\end{equation}

	Recall that $\epsilon=\sqrt{xy}\oplus \sqrt{xz}\oplus \sqrt{yz}$ and it equals zero if and only if $xy\oplus xz\oplus yz=0$. This condition is equivalent to $x=y=0$ or $z=xy/(x\oplus y)$ if $x\ne y$. Thus, we obtain
	\begin{eqnarray}\label{eq-theta2-3}
		\sum_{\substack{x,y,z\in\tei\\ \epsilon=0}}i^{\Tr_1^m(x+3y+2z)} &=& \sum_{\substack{x,y,z\in\tei\\ x=y=0}}i^{\Tr_1^m(x+3y+2z)}+\sum_{\substack{x,y,z\in\tei\\ x\ne y,z=xy/(x\oplus y)}}i^{\Tr_1^m(x+3y+2z)}  \nonumber\\
		&=&\sum_{z\in\tei}i^{\Tr_1^m(2z)}+ \sum_{x,y\in\tei, x\ne y}i^{\Tr_1^m\left(x+3y+2xy/(x\oplus y)\right)} \nonumber\\
		&=&\sum_{x,y\in\tei, x\ne y}i^{\Tr_1^m\left(x+3y+2xy/(x\oplus y)\right)}\nonumber\\
		&=& \sum_{x\in\tei, \mu\in\tei^*}i^{\Tr_1^m(2x+3\mu+2\sqrt{x\mu}+\frac{2x^2}{\mu})} \nonumber\\
		&=& \sum_{x\in\tei, \mu\in\tei^*}i^{\Tr_1^m(2x+3\mu+2x\mu+\frac{2x}{\sqrt{\mu}})}    \nonumber\\
		&=& \sum_{\mu\in\tei^*}i^{\Tr_1^m(3\mu)}\sum_{x\in\tei}i^{\Tr_1^m\left(2(1\oplus \mu\oplus\frac{1}{\sqrt{\mu}})x\right)},
	\end{eqnarray}
	where in the fourth equality we made the substitution $x\oplus y=\mu\in\tei^*$, and the second-to-last equality holds due to the properties of the trace function $\Tr_1^m(x^{2^k}) = \Tr_1^m(x)$ and $\Tr_1^m(x+y)=\Tr_1^m(x)+\Tr_1^m(y)$ for a positive $k$ and $x,y\in\tei$. Observe that $1\oplus \mu\oplus\frac{1}{\sqrt{\mu}}=0$ if and only if $\mu^3\oplus\mu\oplus 1=0$, which is irreducible over $\tei$. Consequently, it has no root in $\tei$ if $3\nmid m$, and it has $3$ roots $\mu, \mu^2, \mu^4$ in $\tei$ if $3\mid m$. Then, by \eqref{eq-theta2-3}, we obtain
	\begin{eqnarray}\label{eq-theta2-4}
		\sum_{\substack{x,y,z\in\tei\\ \epsilon=0}}i^{\Tr_1^m(x+3y+2z)}=
		\left\{ \begin{array}{cl}
			0, &  {\rm if\;}  3\nmid m, \\
			3\cdot 2^mi^{\Tr_1^m(3\mu)}, &  {\rm if\;}  3\mid m
		\end{array}\right.
	\end{eqnarray}
	since $\Tr_1^m(\mu)=\Tr_1^m(\mu^2)=\Tr_1^m(\mu^4)$, where $\mu^3\oplus\mu\oplus 1=0$.

 	Observe that $x^3+x+1$ is irreducible over $\mathbb{F}_2$ and the coefficient of $x^2$ is zero. This implies that $\tr_1^m(\mu)=0$. Using $\mu^3\oplus\mu\oplus 1=0$, we have $\mu^7=\mu\mu^6=\mu(\mu^2\oplus1)=1$. Consequently,  $\mu^{2^j+1}=\mu^2$ if $j=0\,({\rm mod}\, 3)$,  $\mu^{2^j+1}=\mu^3=\mu\oplus1$ if $j=1\,({\rm mod}\, 3)$, and  $\mu^{2^j+1}=\mu^5=\mu^2\oplus\mu\oplus1$ if $j=2\,({\rm mod}\, 3)$. Since $m$ is odd and $3| m$, assume that $m=3(2k+1)$, where $k$ is a non-negative integer. Then, by Lemma \ref{lem-Tr-tr}, we have
	\begin{eqnarray*}
		\Tr_1^m(\mu) &=& \tr_1^m(\mu)+2\oplus_{j=1}^{(m-1)/2}\tr_1^m(\mu^{2^j+1})
		\\ &=& 2\oplus_{j=1}^{3k+1}\tr_1^m(\mu^{2^j+1})
		\\ &=& 2\big(k\cdot\tr_1^m(\mu^2)\oplus(k+1)\cdot\tr_1^m(\mu^3)\oplus k\cdot\tr_1^m(\mu^5)\big)
      	\\ &=& 2((k+1)\oplus k)
       	\\ &=& 2
	\end{eqnarray*}
	since $\tr_1^m(\mu^2)=\tr_1^m(\mu)=0$, $\tr_1^m(\mu^3)=\tr_1^m(\mu\oplus1)=\tr_1^m(1)=1$, and $\tr_1^m(\mu^5)=\tr_1^m(\mu^2\oplus\mu\oplus1)=\tr_1^m(1)=1$ due to $m$ being odd. Therefore, by  \eqref{eq-theta2-2}, \eqref{eq-theta2-3} and \eqref{eq-theta2-4}, we obtain
	\begin{eqnarray*}
		\theta_2=
		\left\{ \begin{array}{cl}
			0, &  {\rm if\;}  3\nmid m, \\
			-3\cdot 2^{3m}, &  {\rm if\;}  3\mid m.
		\end{array}\right.
	\end{eqnarray*}
	This, combined with the definition of $\sigma$ in \eqref{eq-tau-sig}, yields $\theta_2 = 2^{3m}\sigma$.

	By similar reasoning, we can deduce that  $\theta_3 =\theta_5 = \theta_8 =\theta_2 = 2^{3m}\sigma$. Hence, we have	
	\begin{equation*}
		\sum_{a\in\tei^\star}\sum_{b\in\tei}S_+(u)S_+(u+2)S_+(u+1)S_+(u+3) = 8\re(\theta_2) = 8\theta_2 = 2^{3m+3}\sigma.
	\end{equation*}
	This completes the proof.

\end{document}